\documentclass[reqno,11pt]{amsart}
\usepackage{amsmath, latexsym, amsfonts, amssymb, amsthm, amscd}

\numberwithin{equation}{section}

\newtheorem{theorem}{Theorem}[section]

\newtheorem{prop}[theorem]{Proposition}

\newtheorem{rem}[theorem]{Remark}

\newcommand{\un}{\underline}
\newcommand{\ga}{\alpha}

\newcommand{\cB}{{\ensuremath{\mathcal B}} }

\newcommand{\cE}{{\ensuremath{\mathcal E}} }
\newcommand{\cH}{{\ensuremath{\mathcal H}} }
\newcommand{\cJ}{{\ensuremath{\mathcal J}} }

\newcommand{\cS}{{\ensuremath{\mathcal S}} }

\newfont{\indic}{bbmss12}

\def\bo#1{\hbox{{\indic 1}$_{#1}$}}

\newcommand{\hf}{{\frac{1}{2}}}

\newcommand{\E}{{\ensuremath{\mathbb E}} }

\newcommand{\N}{{\ensuremath{\mathbb N}} }

\newcommand{\bbP}{{\ensuremath{\mathbb P}} }

\newcommand{\R}{{\ensuremath{\mathbb R}} }

\newcommand{\myeqnarray}[1]{
  \begingroup
  \jot=#1pt
  \arraycolsep=2pt
  \begin{eqnarray}}

\newcommand{\eeqnarray}{\end{eqnarray}\endgroup}
\newcommand{\ba}{\begin{array}{cc}}
\newcommand{\ea}{\end{array}}

\title{Hot scatterers and tracers for the transfer of heat in collisional dynamics}

\author{Rapha\"el Lefevere}
\address{Laboratoire de Probabilit{\'e}s et Mod\`eles Al\'eatoires (CNRS U.M.R. 7599) \\ Universit{\'e} Paris 7
-- Denis Diderot, U.F.R Math\'ematiques - Case 7012, B\^atiment Chevaleret,75205 PARIS Cedex 13, FRANCE}
 
\email{lefevere\@@math.jussieu.fr}

\author{Lorenzo Zambotti}
\address{Laboratoire de Probabilit{\'e}s et Mod\`eles Al\'eatoires (CNRS U.M.R. 7599) \\ Universit{\'e} Paris 6
-- Pierre et Marie Curie, U.F.R. Math\'ematiques, Case 188, 4 place
Jussieu, 75252 Paris cedex 05, France }
\email{lorenzo.zambotti\@@upmc.fr}

\date{}

\begin{document}

\maketitle

\begin{abstract}We introduce stochastic models for the transport of heat in systems described by local collisional dynamics.  The dynamics consists of tracer particles moving through an array of hot scatterers describing the effect of heat baths at fixed temperatures.  Those models have the structure of Markov renewal processes.  We study their ergodic properties in details and provide a useful formula for the cumulant generating function of the time integrated energy current. We observe that out of thermal equilibrium, the generating function is not analytic.  When the set of temperatures of the scatterers is fixed by the condition that in average no energy is exchanged between the scatterers and the system, different behaviours may arise.  When the tracer particles are allowed to travel freely through the whole array of scatterers, the temperature profile is linear.  If the particles are locked in between scatterers, the temperature profile becomes nonlinear. In both cases, the thermal conductivity is interpreted as a frequency of collision between tracers and scatterers.
\end{abstract}
\section{Introduction.}
\subsection{Lattice Hamiltonian dynamics}

\subsubsection{Smooth interactions}
The study of  conduction of thermal energy by Hamiltonian lattice dynamics has recently known a great deal of activity and numerical as well as analytical results have accumulated. For one-dimensional systems, general lattice Hamiltonian dynamics may be defined in the following way.  Consider $N$ particles  of unit mass  located on a one-dimensional
lattice with local positions and momenta
$(\underline{\mathbf{q}}, \underline{\mathbf{p}}) \equiv
\big\{(\mathbf{q}_i, \mathbf{p}_i)\big\}_{1\leq i\leq N}$, with
$\mathbf{q}_i, \mathbf{p}_i \in \mathbb{R}^d$. The Hamiltonian $H$ takes
the form,
\begin{equation}
H(\underline{\mathbf{p}}, \underline{\mathbf{q}})
= \sum_{i=1}^N \left[\frac{\mathbf{p}_i^2}{2} +V(\mathbf{q}_i)+
U(\mathbf{q}_{i}-\mathbf{q}_{i+1}) \right],
\label{Hamilton}
\end{equation}
where $V$ represents the interaction with an external substrate and $U$
a nearest-neighbor interaction.  Typically, the pinning potentials $V$ and $U$ are smooth and confining, i.e. they grow to infinity when the norm of their argument goes to infinity.  Much effort has been devoted to the study of those systems in the case where the dynamics is a small perturbation of a completely integrable one, namely the case of chains of weakly anharmonic oscillators.  For particles moving in a one-dimensional space, those are described for instance by potentials of the form,
$$
U(x)=\hf\omega^2x^2\,,\,\quad \,V(x)=\hf\nu^2 x^2+\frac{1}{4}\lambda x^4,
$$
with $\lambda$ small .
 The situation may be summarized as follows \cite{Aoki,BricmontKupiainen,LS2,Spohnphonon},  the dynamics is properly described by a Peierls-Boltzmann equation for phonons and under the assumption that this equation holds, the conductivity is finite and may be computed as a function of temperature and microscopic interactions. The temperature dependence is of the form $1/T^2$.

\subsubsection{Collisional dynamics}
Another  type of lattice Hamiltonian dynamics arises in systems where the local dynamics is given by a billiard dynamics \cite{bunilive,GG,GL,Prosen}.  Physically, dynamics of this type can model {\it aerogels}, namely gels from which one has removed the liquid components and replaced them by molecules of gas.
 For instance, one may take in (\ref{Hamilton}) particles moving in one dimension, i.e $d=1$, and take a sequence of interactions $V_k$ and $U_k$,
$$
V_k(x)=f_k\left(\frac{x}{b}\right)\,,\;\;U_k(x)=f_k\left(\frac{x}{a}\right)\,,\;\;
f_k(x)=\frac{x^{2k}}{2k}.
$$
In the limit $k\to\infty$, one obtains,
\begin{eqnarray}
\label{eq: well}
V_{\infty}(x)=\left\{
\begin{array}{l}
  +\infty\; {\rm if}\; |x|>b\\
  0\; {\rm if }\; |x|\leq b
\end{array}
\right. \qquad U_{\infty}(x)=\left\{
\begin{array}{l}
  +\infty\; {\rm if}\; |x|>a\\
  0\; {\rm if }\; |x|\leq a
\end{array}
\right.
\end{eqnarray}
and the dynamics is described by a sequence of ``collisions" between nearest neighbors. When the difference of positions of two neighboring particles reach the parameter $a$, they exchange their velocities and thus their kinetic energies.  This model, originally introduced in \cite{Prosen}, was dubbed in \cite{GL} the {\it complete exchange} model.  One may also simply consider particles moving in square cells located on a one-dimensional lattice.  While remaining confined at all times, the particles collide with their nearest-neighbors through holes in the cells walls. The first type of such models, in which the local dynamics is described by semi-dispersing billiards was introduced in \cite{bunilive} and its thermal transport properties were studied in \cite{GG}.  General collisional models of this type have been introduced in \cite{GL} and may be described as follows.
Formally, the dynamics is described by a Hamiltonian of the form (\ref{Hamilton})  where interaction potentials is equal to zero inside a region $\Omega_U\subset\mathbb{R}^d$ with
smooth boundary $\Lambda$ of dimension $d-1$, and equal to infinity outside.
Likewise, the pinning potential $V$ is assumed to be zero inside a
bounded region $\Omega_V$ and infinity outside, implying that the motion of
a single particle remains confined for all times. The regions $\Omega_U$ and
$\Omega_V$ being specified, the dynamics is equivalent to a billiard in
high dimension.

\subsubsection{Physical observables}
Dynamics described by smooth interactions and collisional dynamics have several qualitative differences and similarities that we briefly point out.
In both cases the interactions between components occur only between nearest neighbors on a lattice and thus the evolution of the local energy may be written as,
\begin{equation}\label{evolen}
E_n(t)-E_n(0)=J_{n-1\to n}([0,t])-J_{n\to n+1}([0,t]).
\end{equation}
In the case of smooth interaction potentials, the {\it time integrated energy current} between sites $n$ and $n+1$ takes the form
$$
J_{n\to n+1}([0,t])=\int_0^t\hf (\mathbf{p}_n(s) + \mathbf{p}_{n+1}(s))\cdot\nabla
U(\mathbf{q}_n(s) - \mathbf{q}_{n+1}(s)) ds
$$
whereas in the case of collisional dynamics,
\begin{eqnarray}
  J_{n\to n+1}([0,t])&=& \hf
  \sum_{0\leq k\leq N_t} \left[ p_{n}^\bot(S^k_n)^2 -
    p^\bot_{n+1}(S^k_n)^2 \right]\label{timeint},
 \end{eqnarray}
where the component of the vector $\mathbf{p}_n$ in the direction of the unit vector
$\widehat{\mathbf{n}}=||q_i-q_{i+1}||^{-1}(q_i-q_{i+1})$ at the time of collision is denoted by $p^\bot_{n}=
\mathbf{p}_n\cdot \widehat{\mathbf{n}}$.
 $N_t$ counts the number of collision up to time $t$ and $(S^k_n)_k$ is the sequence of collision times. Note that in the complete exchange model, $p_n^\bot=p_n$ and the time integrated current between two neighbors is simply the sum of all kinetic energy exchanges between  the particles.
   Assume now that such systems  are thermalized at different temperatures at their boundaries.  In order to understand the transfer of energy from one side to the other, one is interested in the ergodic behavior of the current and in computing
$
\lim_{t\to +\infty}t^{-1}  J_{n\to n+1}([0,t])
$
which gives the average current of energy in the stationary state.  Because of the special form of the time-integrated current (\ref{timeint}), a natural guess to make is that local equilibrium settles in and to assume that this limit is given by
\begin{equation}
\lim_{t\to +\infty}\frac{1}{t}  J_{n\to n+1}([0,t])=\nu \, (T_n-T_{n+1})
\label{Fourier}
\end{equation}
where $\nu=\lim_{t\to\infty} t^{-1}N_t$ is the frequency of collisions between neighbors under local equilibrium conditions.  $T_n=\hf\langle p^2_n\rangle$ is the average kinetic energy of the particles.  The conductivity is thus identified to the frequency of collisions.  Typically, the collisions occur when the particles get near the boundaries of their cell and thus the frequency of collisions is roughly proportional to the  average number of visits to the boundaries per unit time. Because the particle travels freely within its cell, this is proportional to $\sqrt{T_n}$.  Once this relation is taken for granted, then the temperature profile may be computed because, by conservation of energy, the current must be constant throughout the system,
$$
\lim_{t\to +\infty}\frac{1}{t}  J_{(n-1)\to n}([0,t])=\lim_{t\to +\infty}\frac{1}{t}  J_{n\to n+1}([0,t]).
$$
This is equivalent to a finite difference equation for the set of temperatures of the local equilibrium distribution.
 Numerical studies show  that the identification of the conductivity with the frequency of collisions holds true to a very high degree of accuracy in a wide class of collisional dynamics, when the individual particles collide rarely.  In particular, it was shown \cite{GL} that it does not depend on the detailed chaotic properties of the local dynamics as it was originally assumed \cite{GG}.

 Although such a simple way to compute heuristically the heat conductivity does not exist in weakly anharmonic systems, there are similarities between the two dynamics considered in a weakly interacting  regime and written in proper coordinates.   In both cases, a linearized Boltzmann equation describes properly the thermal properties of the systems.   Depending on the context,
each mode or particle behaves as if it was coupled to an ideal stochastic heat bath and the intensity of the coupling depends on the microscopic interactions and yields the thermal conductivity. This quantity is also identified with the frequency of collisions between the components (phonons or particles) in equilibrium.  More precisely in the case of collisional dynamics, each particle moves freely within its own cell and interact with its neighbors as if those were part of infinite thermal bath with fixed temperature.
The Boltzmann type approach is successful in computing theoretically the conductivity from microscopic interactions and local temperature as was checked in numerical simulations \cite{Aoki,GL}.

\subsection{Models and results.}
An important feature of the collisional models is that the evolution of energy occurs at discrete (collision) times  and amounts to an exchange of kinetic energy between neighbors.  The length of the interval of time between two successive collisions depends itself on the kinetic energy of the particle, which fixes the $\sqrt{T}$ dependance of the collision frequency. Our idea is to build and analyze models which are stochastic from the start and share the general structure described by the Boltzmann description sketched above.

Thus, we want to consider dynamics which consists of a mixture of integrable Hamiltonian dynamics and collisions with stochastic heat baths.  The models are made of scatterers described as heat baths and tracer particles transferring energy between those ``hot" scatterers.  The tracers move in a one-dimensional interval in which the scatterers are located on a lattice.  The motion of a tracer is ballistic except when it encounters a scatterer.  At that point its velocity is randomly updated according to a law which depends on the temperature of the scatterer.
The temperature of the scatterers is fixed by the condition that in the stationary state, no energy is exchanged between the scatterers and the tracers.  In spirit,  this is similar to the so-called self-consistent chain of (an)harmonic oscillators \cite{BRV,Bonetto,Bonettoolla}, but in our case it may be naturally interpreted as a condition ensuring that the energy transfer per unit time between scatterers is constant throughout the system.

\subsubsection{Relation to previous works}
Geometrically, the systems we study are analogous to the ones introduced and studied in \cite{EckmannYoung, EckmannYoung2,Mejia}.  The scatterers of our models are similar to the energy storing devices of those models. However, in our case, the dynamics is stochastic from the start and the action of the scatterer models the one  of  a very large system.  As in those models we distinguish two types of dynamics.  Depending on whether the particles are confined or not between the scatterers, we derive a non-linear or linear profile of temperatures for the scatterers.  Those two types of behaviours described by {\it wandering} or {\it confined} tracers seem to be universal in systems described by a collisional dynamics at a microscopic level \cite{EckmannYoung, EckmannYoung2,Mejia,GG,GL,Prosen}.  Among the collisional dynamics described above, two typical examples are given in \cite{Prosen} and \cite{GG,GL}.  In \cite{Prosen}, a detailed numerical analysis of the complete exchange model is provided.  Fourier's law holds and the temperature profile is linear.  In our framework this may be understood as an instance of wandering tracer dynamics.  Indeed, in the complete exchange dynamics, it is not only the energy that is exchanged between neighbors but also the momenta of the particles.  Thus, the dynamics is more similar to the one of particles traveling through the whole system.  In \cite{GG,GL}, models where the particles are confined and exchange only a fraction of their energy are considered. They display a temperature profile identical to the ones of the confined tracers.

\subsubsection{Organisation of the paper}
As stochastic processes, our models are naturally described by {\it Markov renewal processes} and
in section \ref{basicprocess} we consider the simplest dynamics of the type we want to study. We define the stochastic process generated by the free motion of a particle  in a box.  The particle collides with the walls of the box and its velocity is randomly updated according to some fixed probability law. We compute the unique invariant measure for this process. This justifies rigorously the updating rule used in computer simulations \cite{thermal}. In section \ref{tracers}, we define the tracers and scatterers models in full generality and provide an explicit formula for the stationary measure out of equilibrium.  In section \ref{wander}, we describe and analyze the properties of the {\it wandering tracer} model in its simplest version, i.e. when a tracer encounters a scatterer it is deterministically transmitted on the other side of the scatter.  We find that when the temperatures are fixed and the temperatures in the bulk chosen such that the transfer of energy is constant throughout the system, then the temperature profile of the scatterers is linear.  The identity between thermal conductivity and frequency of collisions between tracers and scatterers appears as a natural consequence of the renewal theorem for Markov renewal processes. Next, we study the cumulant generating function of the time-integrated current of energy. We give a rather explicit formula allowing to compute derivatives of any order.  A striking feature is the lack of analyticity of the the generating function.  The origin of this phenomenon may traced back to the presence of particles with arbitrarily low speed. Nevertheless, we are able to show the validity of the Green-Kubo formula for the conductivity. The final section \ref{confined} is devoted to the analysis of the dynamics of confined tracers.   In that case, we find that the condition that there are no exchange of energy between the scatterers and the particles imposes a non-linear profile of temperature.

\section{Basic process.}\label{basicprocess}

We first describe in detail a model which will be the elementary building block of the models we intend to study.
We consider a single particle moving  in the interval $[0,1[$ with a positive velocity. When
the particle reaches $1$, it is absorbed and re-emitted in $0$ with a random positive velocity.

To be more precise, we consider an i.i.d. sequence
$(v_i)_{i=1,2,\ldots}$ such that $v_i>0$ a.s. for all $i$. The particle
is re-emitted from $0$ with speed $v_i$ after its $i$-th collision with the wall
at $1$. The time to reach $1$ again is then $\tau_i:=1/v_i$.

Suppose the particle starts at time $t=0$ at position
$q_0\in[0,1[$ with velocity $p_0>0$. The time of the first collision
with the wall at $1$ is
\[
S_0=S_0(q_0,p_0):=\frac{1-q_0}{p_0},
\]
and the time of the $n$-th collision is
\[
S_n:=S_0+\tau_1+\cdots+\tau_n, \qquad n\geq 1.
\]
We define now a stochastic process $(q_t,p_t)_{t\geq 0}$
with values in $[0,1[\times \R_+$
\[
(q_t,p_t) = F(q,p,t,(\tau_n)_{n\geq 1}) :=
\left\{
\begin{array}{ll}
(q_0+p_0t,p) \qquad {\rm if} \quad  t<S_0,
\\ \\
\left( \frac{t-S_{n-1}}{\tau_n},\frac 1{\tau_n}
\right)  \quad {\rm if} \quad
S_{n-1}\leq t< S_n, \quad n\geq 1.
\end{array}
\right.
\]
We denote by $\cB_b$ the set of all bounded Borel $f:[0,1[\times \R_+\mapsto\R$
and we set
\[
P_tf(q_0,p_0) := \E(f(q_t,p_t)) = \E(f(F(q,p,t,(\tau_n)_{n\geq 1})), \qquad (q_0,p_0)\in[0,1[\times \R_+.
\]
Then it is not difficult to prove that
\begin{prop}\label{markov}
The process $(q_t,p_t)_{t\geq 0}$ is Markov and $(P_t)_{t\geq 0}$
has the semigroup property: $P_{t+s}=P_tP_s$, $t,s\geq 0$.
\end{prop}

\subsection{The invariant measure}

We assume now that
\[
\mu:=\E(\tau_i)=\E(1/v_i)<+\infty,
\]
and that the distribution of $\tau_i$ is {\it non-lattice}, i.e.
there is no $\delta\geq 0$ such that $\bbP(\tau_i\in\delta\N)=1$.
This assumption is not necessary but it simplifies the presentation;
for the applications we have in mind, the distribution of $\tau_i$ has
a density and is therefore always non-lattice.

We denote
the law of $\tau_i$ by $\psi(d\tau)$ and the law of $v_i=1/\tau_i$
by $\phi(du)$.  The law of $\tau_1+\cdots+\tau_n$ is denoted
as usual by the $n$-fold convolution $\psi^{n*}$.
\begin{prop}\label{invmeas}
The only invariant measure on $[0,1[\times\R_+$ of the process $(q_t,p_t, t\geq 0)$ is given by
$\gamma(dq,dp)= dq \, \phi(dp)/\mu p$.
\end{prop}

\begin{proof} For any bounded Borel function
$f:[0,1[\times\R_+\mapsto\R$
\[
\begin{split}
& P_tf(q_0,p_0)=
\\ & = \bo{(t<S_0 )} \, f(q_0+p_0t,p_0)+ \bo{(t\geq S_0)} \, \sum_{n=1}^\infty
\E\left(\bo{(S_{n-1}\leq t< S_n)} \,
f\left(\frac{t-S_{n-1}}{\tau_n}, \frac1{\tau_n}\right)\right)
\\ & = \bo{(t< S_0)} \, f(q_0+p_0t,p_0)+ \bo{(t\geq S_0)} \, \sum_{n=1}^\infty
\int_0^{t-S_0} \psi^{*(n-1)}(ds) \int_s^{+\infty} \psi(d\tau)\,
f\left(\frac{t-S_0-s}{\tau}, \frac1{\tau}\right)
\\ & = \bo{(t< S_0)} \, f(q_0+p_0t,p_0)+ \bo{(t\geq S_0)}
\int_0^{t-S_0} \, U(ds) \int_s^{+\infty} \psi(d\tau)\,
f\left(\frac{t-S_0-s}{\tau}, \frac1{\tau}\right),
\end{split}
\]
where we recall that $S_n=S_0(q,p)+\tau_1+\cdots+\tau_n$ and we set
$\psi^{*0}(ds)=\delta_0(ds)$ and
\[
\int_a^bU(ds) =\sum_{n=1}^\infty \int_a^b \psi^{*(n-1)}(ds)
= \delta_0(ds) + \sum_{n=1}^\infty \int_a^b \psi^{*n}(ds), \quad
0\leq a\leq b.
\]
The {\it renewal measure} $U(ds)$ gives the average number of
collisions in the time interval $ds$. We obtain
\begin{equation}\label{repressem}
 P_tf(q_0,p_0)= \bo{(t< S_0)} \, f(q_0+p_0t,p_0)+ \bo{(t\geq S_0)}
\int_0^{+\infty} \psi(d\tau)\, \int_0^{\tau\wedge (t-S_0)} \, U_t(ds)
\, f\left(\frac{s}{\tau}, \frac1{\tau}\right)
\end{equation}
where $U_t([a,b])=U([t-b,t-a])$ for $0\leq a\leq b\leq t$.
By Blackwell's renewal theorem (see e.g. \cite[Theorem V.4.3]{asmussen}),
$U([t-b,t-a])$ converges to
$(b-a)/\mu$ as $t\to+\infty$. Therefore,
the last expression converges to
\[
\begin{split}
&\int_0^{+\infty} \psi(d\tau) \int_0^\tau \frac{ds}\mu \,
f\left(\frac{s}{\tau}, \frac1{\tau}\right) = \int_0^1
dx \int_0^{+\infty} \frac{\tau\, \psi(d\tau)}\mu
\, f\left(x, \frac1{\tau}\right)
\\ & = \int_{[0,1[\times\R_+} f(x,u) \, \frac{\phi(du)}{\mu\, u}
\, dx =: \int_{[0,1[\times\R_+} f(x,u) \,  \gamma(dx,du).
\end{split}
\]
In other words we have for all bounded Borel $f:[0,1[\times\R_+\mapsto\R$
\[
\lim_{t\to+\infty} P_tf(q_0,p_0) = \int_{[0,1[\times\R_+} f \,  d\gamma,
\qquad \forall \ (q_0,p_0)\in[0,1[\times\R_+.
\]
Since $(P_t)_{t\geq 0}$ is a semigroup, we obtain
\[
\int f \, d\gamma = \lim_{s\to+\infty} P_sf(q_0,p_0) = \lim_{s\to+\infty} P_{t+s}f(q_0,p_0)
= \lim_{s\to+\infty} P_{s}P_tf(q_0,p_0)=\int P_tf \, d\gamma,
\]
i.e. $\gamma$ is invariant for $(P_t)_{t\geq 0}$.
 This convergence result for all initial conditions
$(q_0,p_0)\in[0,1[\times\R_+$ implies that $\gamma(dq,dp)$ is the only invariant measure
of the process.
\end{proof}

\noindent We remark that the invariant measure $\gamma$ satisfies
\begin{equation}\label{gamma}
\int f \, d\gamma = \E\left( f\left(U,\tau^{-1}\right) \, \frac\tau\mu\right),
\end{equation}
where $(U,\tau)$ is independent, $U$ is uniform on $[0,1]$ and
$\tau$ has same law as $\tau_1$.

\subsection{The main example}\label{eqdyn}

The process we study here will be used to build  dynamics where the
updating of the velocities simulate the action of a thermal bath acting
through a wall. Assume that in our set-up one wishes to obtain a process
whose  invariant measure $\gamma(dq,dp)$  is given by
\[
\gamma(dq,dp) = \bo{[0,1]}(q)\, \bo{]0,+\infty[}(p)\, \sqrt{\frac{2\beta}\pi}
\,  e^{-\beta\frac{p^2}2}\, dq\, dp,
\]
where $\beta>0$.
Then from the above results we see that we must choose the
updating distribution
\begin{equation}\label{phibeta}
\phi_\beta(p) = \bo{(p>0)} \, p\,  \sqrt{\frac{2\beta}\pi}\,
e^{-\beta\frac{p^2}2} \, \frac1{\int_{\R_+}
u\,  \sqrt{\frac{2\beta}\pi}\, e^{-\beta\frac{u^2}2}\, du} =
\bo{(p>0)} \, \beta\, p\, e^{-\beta\frac{p^2}2}.
\end{equation}
The distribution of the interarrival time becomes accordingly
\begin{equation}\label{psibeta}
\psi_\beta(d\tau) := \bo{(\tau>0)} \, \frac{\beta}{\tau^3}
\, \exp\left(-\frac\beta{2\tau^2}\right)\, d\tau,
\end{equation}
where we recall that $\phi_\beta$ and $\psi_\beta$ are related by
\[
\int f(v) \, \phi_\beta(v)\, dv =
\int f\left(\frac 1\tau\right) \, \psi_\beta(d\tau).
\]

\subsection{A renewal theory viewpoint}

We recall now some classical results of renewal theory and show
their relevance in our setting.

We consider again the i.i.d. sequence $(\tau_n)_{n\geq 1}$
that we used for the construction of the process
$(q_t,p_t)$. We recall that $\tau_n>0$ a.s., $\E(\tau_n):=\mu<+\infty$
and the distribution of $\tau_i$ is non-lattice.
A (delayed) renewal process associated with $(\tau_n)_{n\geq 1}$
is a sequence
\[
S_0\geq 0, \qquad S_n:=S_{n-1}+\tau_n, \quad n\geq 1,
\]
where $S_0$ is independent of $(\tau_n)_{n\geq 1}$. We define
\[
n_t := \sum_{n=0}^\infty \bo{(S_n\leq t)} =
\#\{n=0,1,\ldots: S_n\leq t\} = \inf\{n: S_n>t\}.
\]
We assume for simplicity that $S_0>0$ and $S_0$ is non-random.
Then we define the processes
\[
B_t:=S_{n_t}-t, \qquad t\geq 0,
\]
\[
A_t:= \left\{ \begin{array}{ll} t-S_{n_t-1}, \quad t\geq S_0, \\ \\
A_0+t, \qquad t<S_0, \end{array} \right.
\]
where $A_0\geq 0$ is non-random and $B_0=S_0>0$ by assumption.
Then it is easy to see that in our setting, for $(q_0,p_0)\in[0,1[\times]0,+\infty[$
\[
A_0=\frac{q_0}{p_0}, \quad B_0=S_0=S_0(q_0,p_0) \, \Longrightarrow \, q_t = \frac{A_t}{A_t+B_t},
\quad p_t = \frac1{A_t+B_t}, \quad \forall \ t\geq 0.
\]
This follows, in particular, from
\[
\frac1{A_t+B_t}=\left\{ \begin{array}{ll} p_0, \qquad \forall \ t< S_0, \\ \\
\tau_{n_t}, \quad \forall \ t\geq S_0, \end{array} \right.
\]
Then we recall the following result (see e.g. \cite[Chapter V]{asmussen})
\begin{prop}\label{renewal}
The processes $(A_t)_{t\geq 0}$, $(B_t)_{t\geq 0}$ and $(A_t,B_t)_{t\geq 0}$ are strong Markov.
Any invariant probability measure of $(A_t)_{t\geq 0}$ is also invariant for $(B_t)_{t\geq 0}$
and must coincide with $x\, \bbP(\tau_1\in dx)/\mu$. The only invariant
probability measure $\nu$ for the process  $(A_t,B_t)_{t\geq 0}$ is given by
\[
\int f\, d\nu = \E\left( f(U\tau,(1-U)\tau) \, \frac\tau\mu\right),
\]
where $(U,\tau)$ is independent, $U$ is uniform on $[0,1]$ and
$\tau$ has same law as $\tau_1$.
\end{prop}
\noindent Since $(q_t,p_t)$ is given by an injective function computed at
$(A_t,B_t)$, then $(q_t,p_t)$ is also strong Markov; moreover, an invariant
probability measure $\gamma(dq,dp)$ for $(q_t,p_t)$ is necessarily given by
\[
\int f\, d\gamma
= \E\left( f\left(\frac{U\tau}{U\tau+(1-U)\tau},\frac1{U\tau+(1-U)\tau}\right) \, \frac\tau\mu\right)
= \E\left( f\left(U,\tau^{-1}\right) \, \frac\tau\mu\right)
\]
which is the same formula found in \eqref{gamma} above.

\subsection{Markov renewal processes}\label{mrt}

A natural generalization of a renewal process is provided by
a {\it Markov renewal process} (see \cite{asmussen}),
which is a pair of stochastic processes
$(X_n,\tau_{n+1})_{n\geq 0}$, such that
\begin{enumerate}
\item $(X_n)_{n\geq 0}$ is a Markov chain in some finite state space $E$
with transition density $(q_{ij})_{i,j\in E}$
\item conditionally on $\cH=\sigma((X_n)_{n\geq 0})$, $(\tau_{n+1})_{n\geq 0}$ is an
independent sequence of positive random variables such that for all $n\geq 1$
\[
\bbP(\tau_n\leq t \, | \, \cH) = \bbP(\tau_n\leq t \, | \, X_n,X_{n+1})
= G_{ij}(t)
\]
on the event $\{X_n=i,X_{n+1}=j\}$, where $(G_{ij})_{i,j\in E}$ is a family of
distribution functions on $]0,+\infty[$.
\end{enumerate}
The sequence of interarrival times $(\tau_{n+1})_{n\geq 0}$ is in
general not i.i.d. However, given the Markov chain $(X_n)_{n\geq 0}$,
$(\tau_{n+1})_{n\geq 0}$ is independent and the law of $\tau_{n+1}$
depends on $X_n$ and $X_{n+1}$.
This kind of process will be the main mathematical tool in the concrete models
we analyze in the following sections.

\section{Tracers and hot scatterers}\label{tracers}
\noindent We consider a gas of $M$ non-interacting tracer particles moving
through a one-dimensional lattice of scatterers $w_n$, $n=1,\ldots,N$.
In between the scatterers, the tracers move with constant speed in the
boxes $I_i$,
\begin{equation}
I_n=\left[n-1,n\right],\qquad n=1,\ldots,N
\end{equation}
When a tracer encounters the scatterer $n$, it is absorbed and re-emitted
on either side according to a certain probability distribution
with a random velocity $p$ distributed according to the law
\begin{equation}
\phi^{\pm}_{\beta_n}(p)=p^{\pm} \, \beta_n \, e^{-{\beta_n} \frac{p^2}{2}},
\end{equation}
see \eqref{phibeta}.
The state space describing the motion of the particle is thus the cartesian product of the positions space $I=[0,N]$ and
velocity space $\R^*=\R\backslash\{0\}$, $\Omega=I\times\R^*$.
At the extremities of the system
the sign of the velocity is reversed but the particle bounces back with a random velocity
distributed according  to the same law with parameter $\beta_L$ and $\beta_R$.
This dynamics defines a Markov renewal process which we describe now more formally using the same notations as above.

For notational simplicity, we define the dynamics for a single tracer particle,
the extension to $M$ particles is straightforward.
The particle moves in the interval
$I:=[0,N]$, which is split into $N$ subintervals of equal length:
$I_n:=[n-1,n]$, $n=1,\ldots,N$. At time $t=0$,
the particle starts at position $q_0\in\, ]0,N[$ with speed
$p_0\in\R^*:=\R\backslash\{0\}$ and we define $(n_0,\sigma_0)\in E$ as follows:
$n_0:=\lfloor q+({\rm sign}(p_0)+1)/2\rfloor\in\{0,1,\ldots,N\}$,
where $\lfloor\cdot\rfloor$ denotes the integer part; in other words $n_0$ is such that
\[
\left\{\begin{array}{ll}
n_0-1< q_0 \leq n_0
\qquad {\rm if} \quad {\rm sign}(p_0)=+1,
\\ n_0\leq q_0<n_0+1 \qquad {\rm if} \quad {\rm sign}(p_0)=-1.
\end{array} \right.
\]
and denotes the first scatterer the particle hits. We define moreover
\[
\sigma_0:=\left\{\begin{array}{ll}
{\rm sign}(p_0),
\qquad {\rm if} \quad n_0\in\{1,\ldots,N-1\},
\\ -{\rm sign}(p_0), \qquad {\rm if} \quad n_0\in\{0,N\}.
\end{array} \right.
\]
In other words, if the first scatterer the particle hits is at the 
boundary $\{0,N\}$ of the system, then the particle will be reflected
at the first hitting.

We suppose that the sequence $X_k=(n_k,\sigma_k)$
of scatterers visited by the particle and signs of the
velocity is a Markov chain on
\[
E:=\{(n,\sigma), \ n=1,\ldots,N-1, \ \sigma=\pm1\} \cup \{(0,+1), (N,-1)\},
\]
with initial state $(n_0,\sigma_0)$
and with an irreducible probability transition matrix on $E$
such that
\[
q_{(n,\sigma),(n',\sigma')} = \left\{ \begin{array}{ll}
1 \qquad {\rm if} \quad (n,\sigma)=(1,-1)  \ {\rm and} \  (n',\sigma')=
(0,+1)
\\
1 \qquad {\rm if} \quad (n,\sigma)=(N-1,+1)  \ {\rm and} \  (n',\sigma')=
(N,-1)
\\ 0 \qquad {\rm if} \quad n\in\{1,\ldots,N-1\} \quad {\rm and} \quad
n'-n\ne\sigma'.
\end{array} \right.
\]
The first two conditions mean that the tracer is always reflected
at $n=0$ and $n=N$.
The last condition means that the new sign $\sigma'$ gives the next
scatterer $n'$ visited:
if $\sigma'=+1$ then $n'=n+1$, if $\sigma'=-1$ then $n'=n-1$.
The irreducibility assumption gives the existence of a unique
invariant probability measure that we call $(\nu_\alpha)_{\alpha\in E}$.

We now define the time the particle takes between two
subsequent visits to the scatterers.
Conditionally on $\cH=\sigma((X_k)_{k\geq 0})$, the sequence $(\tau_{k})_{k\geq 1}$
is independent with distribution defined by
\begin{equation}\label{tau}
\bbP(\tau_k\in d\tau \, | \, \cH) = \bbP(\tau_k\in d\tau \, | \, X_{k-1})
= \frac{\beta_{n}}{\tau^3} \,
\exp\left(-\frac{\beta_{n}}{2\tau^2}\right)
\, \bo{(\tau>0)} \, d\tau =: \psi_{n}(d\tau)
\end{equation}
on the event $\{X_{k-1}=(n,\sigma)\}$, where $\beta_0,\ldots,\beta_N\in\R_+$,
see \eqref{psibeta}.

We consider now the Markov chain $(X_k)_{k\geq 0}$ with initial
state $X_0=(n,\sigma)$ and the associated sequence $(\tau_k)_{k\geq 1}$.
The time of the first collision with a wall is
\[
S_0=S_0(q_0,p_0):=\frac{n_0-q_0}{p_0}>0,
\]
and the time of the $k$-th collision with one of the scatterers is
\[
S_k:=S_0+\tau_1+\cdots+\tau_k, \qquad k\geq 1.
\]
Before time $S_0$, the particle moves with uniform velocity
$p$. Between time $S_{k-1}$ and time $S_k$, the particle
moves with uniform velocity $\frac{\sigma_k}{\tau_k}$ and
$(S_k)_{k\geq 0}$ is the sequence of times when $q_t\in
\{0,\ldots,N\}$. In particular we define the sequence of incoming velocity
$v_k$ at time $S_k$
\begin{equation}\label{v}
v_0:=p_0, \qquad v_k:=\frac{\sigma_k}{\tau_k}, \quad k\geq 1.
\end{equation}
We define precisely the stochastic process $(q_t,p_t)_{t\geq 0}$
with values in $[0,N]\times \R^*$
\begin{equation}\label{qp}
(q_t,p_t) :=
\left\{
\begin{array}{ll}
(q_0+p_0t,p_0) \qquad {\rm if} \quad  t<S_0,
\\ \\
\left( n_{k-1}
+\frac{\sigma_k}{\tau_k}(t-S_{k-1}),
\frac {\sigma_k}{\tau_k}
\right)  \ \ {\rm if} \ \
S_{k-1}\leq t< S_k, \ k\geq 1,
\end{array}
\right.
\end{equation}
and we use the notation $X_k=(n_k,\sigma_k)$, $X_{k-1}=(n_{k-1},\sigma_{k-1})$.
Then, in analogy with Propositions \ref{markov} and \ref{invmeas}
we have the result
\begin{prop}\label{markovren}
The process $(q_t,p_t)_{t\geq 0}$ is Markov and its only invariant measure
on $[0,1]\times\R^*$ is given by
\begin{equation}\label{selfmu0}
\begin{split}
\gamma(dq,dp) = \frac1{Z_N}\sum_{n=1}^N \bo{I_n}(q) \sum_{\sigma=\pm1}
& \left( \nu_{(n-1,\sigma)}\,
q_{(n-1,\sigma),(n,+1)} \, \bo{(p>0)}\, \beta_{n-1}\, e^{-\beta_{n-1}\frac{p^2}{2}} \, + \right.
\\ & \left. + \, \nu_{(n,\sigma)}\, q_{(n,\sigma),(n-1,-1)} \, \bo{(p<0)}
\, \beta_{n}\, e^{-\beta_{n}\frac{p^2}{2}}\right) dq \, dp
\end{split}
\end{equation}
where $Z_N=\sqrt{\frac\pi 2}
\sum_{(n,\sigma)\in E} \nu_{(n,\sigma)}\, \sqrt{\beta_{n}}$.
\end{prop}
\begin{proof}
The proof is very similar to that of Proposition \ref{invmeas}.
The renewal measures we have to consider is
\[
U_{\alpha,\alpha'}(t) = \sum_{k=1}^\infty \bbP_{\alpha}
\left(S_k\leq t, \ X_k=\alpha' \right),
\]
where under $\bbP_{\alpha}$ the process $X_k$ has the law
described above with $X_0=\alpha=(i,\sigma)$ a.s. Then the analog of formula \eqref{repressem} is
\[
\begin{split}
& P_tf(q_0,p_0) = \bo{(t< S_0)} \, f(q_0+p_0t,p_0) \\ & + \bo{(t\geq S_0)} \sum_{\alpha',\alpha''\in E}
\int_0^{+\infty} q_{\alpha',\alpha''} \,
\psi_{n'}(d\tau)\, \int_0^{\tau\wedge (t-S_0)} \, U^t_{\alpha,\alpha'}(ds)
\, f\left(n'+\frac{\sigma''s}{\tau},
\frac {\sigma''}{\tau}\right)
\end{split}
\]
%
where we use the notation $\alpha=(n,\sigma)$, $\alpha'=(n',\sigma')$, $\alpha''=(n'',\sigma'')$
and the measure $U^t_{\alpha,\alpha'}(ds)$ are defined by
\[
U^t_{\alpha,\beta}([a,b])=U_{\alpha,\beta}(t-b)-U_{\alpha,\beta}(t-a),
\quad 0\leq a\leq b\leq t, \quad \alpha,\beta\in E.
\]
Since the analog of Blackwell's Theorem holds also for
Markov renewal processes, see e.g. \cite[Theorem VII.4.3]{asmussen},
we obtain
\[
\lim_{t\to+\infty} U^t_{\alpha,\beta}([a,b]) = (b-a)\, \frac{\nu_{\beta}}{\overline\mu}
\]
where $\nu$ is the unique probability invariant measure on $E$ of the Markov
chain $(X_k)_{k\geq 0}$ and
\[
\overline\mu := \sum_{(n,\sigma)\in E} \nu_{(n,\sigma)} \, \int_{\R_+} \tau\,
\psi_{n}(d\tau) = \sqrt{\frac\pi 2}
\sum_{(n,\sigma)\in E} \nu_{(n,\sigma)}\, \sqrt{\beta_{n}}.
\]
Therefore
\[
\begin{split}
\lim_{t\to+\infty} P_tf(q_0,p_0) &= \sum_{\alpha',\alpha''\in E}
\int_0^{+\infty} q_{\alpha',\alpha''} \,\psi_{n'}(d\tau) \int_0^\tau ds\,
\frac{\nu_{\alpha'}}{\overline\mu} \,
f\left(n'+\frac{\sigma''s}{\tau},
\frac {\sigma''}{\tau}\right) \\ & = \frac1{\overline\mu}
\sum_{\alpha',\alpha''\in E} \nu_{\alpha'} \, q_{\alpha',\alpha''} \int_0^1
dx \int_0^{+\infty} \tau\, \psi_{n'}(d\tau)
\, f\left(n'+\sigma''x, \frac{\sigma''}{\tau}\right)
\\ & = \int_{[0,1]\times\R} f(q,p) \,  \gamma(dq,dp).
\end{split}
\]
Arguing as in the proof of Proposition \ref{markov}, we conclude.
\end{proof}
\begin{rem}{\rm
Notice that the invariant measure is always explicit, although in general
the process $(q_t,p_t)_{t\geq 0}$ is {\it not} reversible.
}
\end{rem}

\section{Wandering tracers}\label{wander}
\subsection{Generalities and physical observables.}
In the first model that we study, when a tracer reaches a scatterer $n\in\{1,\ldots,N-1\}$,
it is absorbed on one side and re-emitted on the other side
with a random velocity distributed according to a law determined by the temperature of the scatterers.
The sign of the velocity changes when and only when the tracer reaches the scatterers $0$ or $N$.
The transition matrix of the underlying Markov chain is
\[
q_{(n,\sigma),(n',\sigma')} = \left\{ \begin{array}{ll}
1 \qquad {\rm if} \quad n'=n+\sigma\notin\{0,N\} \ {\rm and} \  \sigma=\sigma'
\\
1 \qquad {\rm if} \quad (n,\sigma)=(1,-1)  \ {\rm and} \  (n',\sigma')=
(0,+1)
\\
1 \qquad {\rm if} \quad (n,\sigma)=(N-1,+1)  \ {\rm and} \  (n',\sigma')=
(N,-1)
\\ 0 \qquad {\rm otherwise}.
\end{array} \right.
\]
The associated
invariant probability measure given is the uniform distribution on $E$.
In fact, in this case the Markov chain moves
deterministically as follows: $X_k=(n_k,\sigma_k)$ where
\begin{equation}\label{explicit}\left\{
\begin{array}{ll}
n_k & \, = f\left( |n_0+\sigma_0k| \, {\rm mod} \, 2N \right), \qquad f(i) := N-|N-i|, \quad
i=0,\ldots,2N
\\ \\ \sigma_k & \, = \sigma_0 \, (-1)^{\left\lfloor ({n+\sigma_0(n_0-N)+N})/N\right\rfloor }, \qquad k\geq 1.
\end{array}
\right.
\end{equation}
In particular, we have the following periodicity
\begin{equation}\label{period}
(X_k)_{k\geq 0} \stackrel{d}{=} (X_{k+2N})_{k\geq 0} \quad {\rm under} \quad \bbP_{(n_0,\sigma_0)}
\end{equation}
where $\stackrel{d}{=}$ denotes equality in distribution.
Proposition \ref{markovren} becomes
\begin{prop}\label{markovren2}
The process $(q_t,p_t)_{t\geq 0}$ is Markov and its only invariant measure
on $[0,1]\times\R^*$ is given by
\begin{equation}\label{selfmu}
\gamma(dq,dp)= \frac1{Z_N}\sum_{n=1}^N \bo{I_n}(q)\,
\left(\bo{(p>0)}\, \beta_{n-1}\, e^{-\beta_{n-1}\frac{p^2}{2}}+\bo{(p<0)}
\, \beta_{n}\, e^{-\beta_{n}\frac{p^2}{2}}\right) dq \, dp
\end{equation}
where $Z_N=\sqrt{\frac\pi 2}
\sum_{n=1}^{N} \left(\sqrt{\beta_{n-1}}+\sqrt{\beta_n}\right)$.
\end{prop}
\noindent
The first useful result is the computation of the asymptotic frequency of collision
of a tracer with a fixed scatterer.
\begin{prop}\label{N_t}
For $n\in\{0,\ldots,N\}$ set $\phi_{n,0}:=\inf\{\ell\geq 0: \,
n_\ell=n\}$,
\[
\phi_{n,k+1}:=\inf\{\ell>\phi_{n,k}: \,
n_\ell=n\}, \qquad k\geq 0
\]
and
\begin{equation}\label{defN_t}
N_t^n := \sum_{k=1}^\infty \bo{(S_{\phi_{n,k}}\leq t)}
, \qquad \hat N_t^n := \sum_{k=1}^\infty 2\, \bo{(S_{\phi_{n,2k}}\leq t)}
, \qquad t\geq 0.
\end{equation}
Then for any initial condition $(q_0,p_0)$, $\bbP_{(q_0,p_0)}$-a.s.
\begin{equation}\label{eqN_t}
\lim_{t\to+\infty} \frac{N_t^n}t = \lim_{t\to+\infty} \frac{\hat N_t^n}t
=\frac 2{Z_N}, \quad {\rm if} \quad n\in\{1,\ldots,N-1\},
\end{equation}
\begin{equation}\label{eqN_t2}
\lim_{t\to+\infty} \frac{N_t^n}t 
=\frac 1{Z_N}, \quad {\rm if} \quad n\in\{0,N\}.
\end{equation}
\end{prop}
\begin{proof}
Let $n\in\{1,\ldots,N-1\}$. By \eqref{period}, the sequence
$(S_{\phi_{n,2(k+1)}}-S_{\phi_{n,2k}})_{k\geq 0}$ is i.i.d.
and therefore we have by the renewal theorem
\[
\lim_{t\to+\infty} \frac{\hat N_t^n}t = \frac2{\E(S_{\phi_{n,2}}-S_{\phi_{n,0}})}
= \sqrt\frac2\pi \,\frac2{\sum_{i=1}^N (\sqrt{\beta_{i-1}}+\sqrt{\beta_{n}})}.
\]
Since $|\hat N_t^n-N_t^n|\leq 1$, we conclude.
\end{proof}

We next identify the physical quantities of interest. The energy
exchanged between the scatterer $n$ and a particle
during a time interval $[0,t]$ is given by
\[
E_{n}([0,t]):=\hf\sum_{k\geq 0:\, S_k\leq t}\left(v^2_{k+1}-v^2_{k}\right)
\bo{(n_{k}=n)},
\]
recall that, by \eqref{v} and \eqref{qp}, $v_{k}$ and $v_{k+1}$ are respectively the incoming
and the outcoming velocity at time $S_k$.
The total entropy flow due to the exchange of energy between the scatterers and a particle is given by
\begin{equation}\label{entropyflow}
S_n([0,t]):= -\frac{E_{n}([0,t])}{T_n},
\qquad S([0,t]):=\sum_{n=0}^NS_n([0,t]).
\end{equation}
The energy exchanged between scatterers $n$ and $(n+1)$
during a time interval $[0,t]$ is
\[
J_{n\to n+1}([0,t]):=\hf\sum_{k\geq 1:\, S_k\leq t}v^2_{k}
\left( \bo{(n_{k-1}=n, \ \sigma_{k-1}=1)} - \bo{(n_{k-1}=n+1,\  \sigma_{k-1}=-1)} \right).
\]
\noindent We define the energy flow per unit time in the stationary state by
\begin{equation}\label{en}
\cE_n:=\lim_{t\to+\infty} \frac1t\, E_n([0,t]).
\end{equation}
Similarly, the entropy flow per unit time is given by
\begin{equation}\label{sn}
\cS_n:=\lim_{t\to+\infty} \frac1t \, S_n([0,t]),
\qquad \cS:=\sum_{n=0}^N\cS_n
\end{equation}
and the current of energy between scatterers $w_n$ and $w_{n+1}$ is given by the transfer of energy per unit time,
\begin{equation}\label{jn}
\cJ_n:=\lim_{t\to+\infty} \frac1t\, J_{n\to n+1}([0,t]).
\end{equation}
\begin{prop}\label{limit} The limits in \eqref{en}, \eqref{sn} and \eqref{jn} exist $\bbP_{(q_0,p_0)}$ a.s.
and for all $n=1,\ldots,N-1$,
\begin{equation}\label{cE}
\cE_n= \frac{2T_n-T_{n-1}-T_{n+1}}{Z_N}\quad{\rm and}\quad
\cE_0=\frac{T_0-T_1}{Z_N},\quad \cE_N=\frac{T_N-T_{N-1}}{Z_N},
\end{equation}
\begin{equation}\label{cJ}
\cJ_n= \frac{T_n-T_{n+1}}{Z_N},  \qquad \cS=\frac1{Z_N}
\sum_{n=0}^{N-1}\frac{(T_n-T_{n+1})^2}{T_nT_{n+1}}\geq 0.
\end{equation}
\end{prop}
\begin{proof}
Setting $Y_k:=(X_{i+2Nk},i=0,\ldots,2N-1)\in E^{2N}$ then
by \eqref{period} and the Markov property we have that
$(Y_k)_{k\geq 0}$ forms an i.i.d. sequence.
Using the notation \eqref{defN_t} we can write
\[
E_{n}([0,t]) =\hf \sum_{k\geq 0:\, S_k\leq t}\left(v^2_{k+1}-v^2_{k}\right)
\bo{(n_k=n)} =\hf \sum_{k=0}^{N^n_t} \left(v^2_{\phi_{n,k}+1}-v^2_{\phi_{n,k}}\right)
\]
and define
\[
e_k^n := \hf \left(v^2_{\phi_{n,k}+1}-v^2_{\phi_{n,k}}\right), \quad
\hat E_{n}([0,t]) =\sum_{k=0}^{\hat N^n_t/2} \left(e_{2k+1}^n+e_{2k}^n\right).
\]
Notice that $e_k^n$ is the exchange of energy between the scatterer $n$
and the particle at the $(k+1)$th passage of the particle by the scatterer $n$.
Since $(Y_k)_{k\geq 0}$ is i.i.d., then
$(e_{2k+1}^n+e_{2k}^n)_{k\geq 0}$ is also i.i.d. and we obtain by \eqref{eqN_t}
and by the law of large numbers for $n\in\{1,\ldots,N-1\}$
\[
\lim_{t\to+\infty} \frac1t\, \hat E_{n}([0,t]) =
\lim_{t\to+\infty} \frac{\hat N^n_t}{2t}\, \frac2{\hat N^n_t}
\, \hat E_{n}([0,t])= \frac1{Z_N} \, \E_{(q_0,p_0)}(e_1^n+e^n_0).
\]
Now, over a period of $2N$ transitions of $(X_k)_{k\geq 0}$,
each scatterer $n\in\{1,\ldots,N-1\}$ is visited twice, once coming
from the right and once from the left. Therefore
\[
\E_{(q_0,p_0)}(e_1^n+e^n_0) = 2T_n-T_{n-1}-T_{n+1}, \qquad \forall\
n\in\{1,\ldots,N-1\}.
\]
Notice now that $|\hat E_{n}([0,t])-E_{n}([0,t])|\leq
|e^n_{N^n_t}|=:W_t$. It can be
seen that $W_t/t\to 0$ and therefore we obtain the first relation of
\eqref{cE}.

If now $n\in\{0,N\}$, by \eqref{eqN_t2}
\[
\lim_{t\to+\infty} \frac1t\, E_{n}([0,t]) =
\lim_{t\to+\infty} \frac{N^n_t}t\, \frac1{N^n_t}
\, E_{n}([0,t])= \hf\, \frac1{Z_N} \, \E_{(q_0,p_0)}\left(e^n_0\right).
\]
In particular
\[
\hf\,\E_{(q_0,p_0)}\left(e^0_0\right) = T_0-T_{1},
\qquad \hf\,\E_{(q_0,p_0)}\left(e^N_0\right) = T_N-T_{N-1},
\]
and the proof of \eqref{cE} is complete.
The proof of the first relation in \eqref{cJ} is similar. The second relation follows  from
\eqref{cE} and a summation by parts:
$$
Z_N \, \cS =-
\sum_{n=0}^{N-1}(T_n-T_{n+1})\left(\frac{1}{T_{n}}-\frac{1}{T_{n+1}}\right)
=\sum_{n=0}^{N-1}\frac{(T_n-T_{n+1})^2}{T_nT_{n+1}}.
$$
\end{proof}

\subsection{Self-consistency, Fourier's law and temperature profile.}
\noindent
From \eqref{cE} and \eqref{cJ}, we have the obvious result
\begin{prop}{\bf (Self-consistency condition)}
The only collection  $(T_n)_{n=0,\ldots,N}$
such that
$$
 \cE_n=0,\; n=1,\ldots,N-1
$$
with $T_0=T_L$ and $T_N=T_R$ is
\begin{equation}\label{profile}
T_n=T_L+\frac{n}{N}(T_R-T_L), \qquad n=0,\ldots,N.
\end{equation}
The entropy flow per unit time $\cS$ is equal to $0$ if and only if
$T_L=T_R$.
\end{prop}
\noindent Note that the condition on the exchange of energy  is imposed only for the scatterers.  In contrast, when $T_L\neq T_R$ the tracer will always exchange energy with the boundary walls.

\noindent Let us consider now $M_N$ (to be fixed) non-interacting tracer particles described by their momenta and positions $(\un p,\un q)=(p_i,q_i)_{1\leq i\leq M_N}$ and moving through the array of scatterers. As the motions of the tracers are independent, the generalization is straightforward. The corresponding stationary measure is given by
\begin{equation}
\gamma^{M_N}(d\un q,d\un p)=\prod_{i=1}^{M_N}\gamma(dq_i,dp_i)
\end{equation}
and the total average current between scatterers $w_n$ and $w_{n+1}$ is the
sum of the contribution of each particle in (\ref{cJ})
\begin{equation}
\cJ^{M_N}_{n} = {M_N}\, \frac{T_{n}-T_{n+1}}{Z_N}.
\label{totalaverage}
\end{equation}
The total rate of energy exchanged between the scatterer $w_n$ and the tracers in the stationary state is given by
\begin{equation}\label{totenergyrate}
\cE^{M_N}_n={M_N} \, \frac{2 T_n-T_{n-1}-T_{n+1}}{Z_N}.
\end{equation}
Thus, if the the self-consistency condition is imposed and the temperatures of the scatterers is given by (\ref{profile}), then, by (\ref{totalaverage}), one has
\begin{equation}
\cJ^{M_N}_{n}=-\frac{M_N}{N {Z_N}}\, (T_R-T_L).
\end{equation}
The local conductivity is defined as the ratio of the average current of energy to the local temperature gradient, namely,
\begin{equation}
\kappa_n\equiv\lim_{N\rightarrow\infty}\frac{\cJ^{M_N}_{n}}{T_n-T_{n+1}} =
\lim_{N\rightarrow\infty} \frac {M_N}{Z_N}.
\end{equation}
If the temperature profile is given by (\ref{profile}), then  one may compute the
explicit asymptotic behavior of ${Z_N}$ in the large $N$ limit
\begin{equation}
\lim_{N\rightarrow\infty} \frac{{Z_N}}{N}={\sqrt{2\pi} }\int_0^1 \frac{dx}{(T_L+x(T_R-T_L))^\hf}=
\frac{2\sqrt{2\pi}}{T_R^\hf+T_L^\hf }.
\end{equation}
Thus for a number of tracers $M_N=o(N)$, we have $\kappa_n=0$.  This is because when the number of scatterers increases, the proportion of time that a given tracer spends carrying energy from scatterer $w_n$ to $w_{n+1}$ goes to zero simply because the tracer must go back and forth in a larger and larger system.  However, we see that if we take as many tracer particles as scatterers, namely $M_N=N$, then Fourier's law holds, i.e the conductivity is finite.  Its value is given by
\begin{equation}\label{GK0}
\kappa_n=\frac{T_R^\hf+T_L^\hf}{2\sqrt{2\pi}}.
\end{equation}
Notice that $\kappa_n$ does not depend on $n$. In particular, if $T_L=T_R=T$
\begin{equation}\label{GK}
\kappa_n={\sqrt\frac{T}{2\pi}}.
\end{equation}

\subsection{Cumulant generating function and
the Gallavotti-Cohen symmetry relation}\label{wanderingGC}

We denote ${\cB}=(\beta_0,\ldots,\beta_N)\in\R_+^{N+1}$.
We do not need in this section to assume that $T_i:=\beta_i^{-1}$
satisfy  \eqref{profile}. Fix $n\in\{0,\ldots,N-1\}$.
We are going to compute and study the properties of the cumulant generating function,
\begin{equation}\label{freeener}
f_n(\lambda,\cB):=
\lim_{t\to+\infty} \frac1t \, \log\E\left(\exp\left(-\lambda J_n([0,t])\right)\right),
\qquad \forall \, \lambda\in\, ]-\beta_n,\beta_{n+1}[.
\end{equation}
We define $\Delta_n(\alpha):=\left( \bo{(i=n, \ \sigma=1)} -
\bo{(i=n+1,\ \sigma=-1)}\right)$  for $\alpha=(i,\sigma)\in E$.
For $\lambda\in \, ]-\beta_n,\beta_{n+1}[$ and $\epsilon\geq 0$, we define
\begin{equation}\label{calpha}
C_n(\alpha,\lambda,\epsilon):=\beta_{i}
\int_0^{+\infty} v\, e^{- \frac\epsilon v -({\beta_{i}}+\lambda\, \Delta_n(\alpha))\, \frac{v^2}2 } \, dv,
\qquad \alpha=(i,\sigma)\in E.
\end{equation}
and the function $F_n$, which is crucial in the computation of $f_n(\lambda,\cB)$,
\[
F_n(\lambda,\epsilon,\cB):= \prod_{\alpha\in E} C_n(\alpha,\lambda,\epsilon).
\]
$F_n$ will be identified with the spectral radius of some matrix.

We anticipate a striking feature of our model: for $\beta_{n}\ne\beta_{n+1}$,
the cumulant generating function $f_n(\cdot,\cB)$ is {\it not}
analytic around $\lambda=0$. In particular,
$f_n(\cdot,\cB)>0$ for $\lambda$ in a left neighborhood of $0$, while
$f_n(\cdot,\cB)=0$ in a right  neighborhood of $0$. See Remark \ref{analytic}
and section \ref{fluctu} below for further discussions.
\begin{prop}\label{mgf} If $\beta_{n}\leq \beta_{n+1}$ then
$\forall \,\lambda\in\,]-\beta_n,0[ \, \cup \,
]\beta_{n+1}-\beta_n,\beta_{n+1}[$, $f_n(\lambda,\cB)$ is given by the unique solution
$\epsilon_0>0$ to the equation
$$
F_n(\lambda,\epsilon_0,\cB)=1.
$$
If $\lambda\in \, [0,\beta_{n+1}-\beta_n]$, then $f_n(\lambda,\cB)=0$.
The function $f_n(\cdot,\cB)$ is convex and continuous
over $]-\beta_n,\beta_{n+1}[$ and satisfies the Gallavotti-Cohen symmetry relation
\begin{equation}\label{GC}
f_n(\lambda,\cB)=f_n(\beta_{n+1}-\beta_n-\lambda,\cB).
\end{equation}
\end{prop}
\begin{proof}
We call a family of $\sigma$-finite measures $F_{\ga,\ga'}(d\tau)$ on $[0,+\infty[$,
indexed by $(\ga,\ga')\in E\times E$, a {\it kernel}.
Given two kernels $F$ and $G$ we define
their convolution $F*G$ as the kernel
\begin{equation} \label{eq:convolution}
(F*G)_{\ga,\ga'}([0,t]) := \sum_{\gamma\in E} \int_0^t F_{\ga,\gamma}(d\tau) \int_0^{t-\tau}
G_{\gamma,\ga'}(ds).
\end{equation}
We can reduce to the case of $q_0\in\{0,\ldots,N\}$, so that $S_0=0$. We set
\[
\begin{split}
Z_{\alpha,\ga'}(t) & := \E_\alpha\left(\exp\left(-\lambda J_n([0,t])\right)\,
\bo{(X_{N_t}=\ga')}\right)
\\ & = \E_\alpha\left(\exp\left(-\frac\lambda2 \sum_{k=1}^{N_t}v^2_{k} \,
\Delta_n(X_{k-1})\right) \bo{(X_{N_t}=\ga')}\right)
\end{split}
\]
where we recall that $\Delta_n(\alpha):=\left( \bo{(i=n, \ \sigma=1)} -
\bo{(m=i+1,\ \sigma=-1)}\right)$ for $\alpha=(i,\sigma)$, and
\[
N_t := 
\sum_{k=1}^\infty \bo{(S_k\leq t)}.
\]
Let us also set for $\alpha=(i,\sigma)$ and $\alpha'=(i',\sigma')$
\begin{eqnarray}\label{haydn}
\nonumber M_{\alpha,\alpha'}(d\tau) & := & q_{\alpha,\alpha'} \, \psi_{i}(d\tau) \,
\exp\left(-\frac{\lambda\, \Delta_n(\alpha)}{2\, \tau^2}
\right) \\ & = & q_{\alpha,\alpha'} \, \frac{\beta_{i}}{\tau^3} \,
\exp\left(-({\beta_{i}}+\lambda\, \Delta_n(\alpha))\, \frac1{2\tau^2}\right)
\, \bo{(\tau>0)} \, d\tau.
\end{eqnarray}
By summing over all possible values of $N_t$ we obtain
\[
\begin{split}
Z_{\alpha,\ga'}(t) & =  \bbP_\alpha(\tau_1>t)\, \delta_{\alpha,\ga'} +
\sum_{\ell=1}^\infty\E_\alpha\left(\bo{(N_t=\ell)}\,
\exp\left(-\frac\lambda2 \sum_{k=1}^{\ell}v^2_{k} \,
\Delta_n(X_k)\right) \bo{(X_{\ell}=\ga')}\right)
\\ & = \bbP_\alpha(\tau_1>t)\, \delta_{\alpha,\ga'} +\sum_{\ell=1}^\infty M^{\ell*}_{\alpha,\ga'}([0,t]),
\end{split}
\]
where $M^{\ell*}$ is the convolution of $M$ with itself $\ell$ times.
Let us set for $\epsilon\geq 0$
\[
\begin{split}
B^\epsilon_{\alpha,\alpha'} & := \int_0^{+\infty} e^{-\epsilon\tau}\, M_{\alpha,\alpha'}(d\tau)
=q_{\alpha,\alpha'} \, \beta_{i}
\int_0^{+\infty} v\, \exp\left(-\frac\epsilon v -({\beta_{i}}+\lambda\, \Delta_n(\alpha))\, \frac{v^2}2\right) dv
\\ & = q_{\alpha,\alpha'} \, C_\alpha(\lambda,\epsilon),
\end{split}
\]
where $C_\alpha$ is defined in \eqref{calpha}, and let us denote
\[
B_{\alpha,\alpha'}=B^0_{\alpha,\alpha'} = q_{\alpha,\alpha'} \, \beta_{i}
\int_0^{+\infty} v\, \exp\left(-({\beta_{i}}+\lambda\, \Delta_n(\alpha))\, \frac{v^2}2\right) dv =
\frac{q_{\alpha,\alpha'} \, \beta_{i}}{{\beta_{i}}+\lambda\, \Delta_n(\alpha)}.
\]
Because of the explicit form of the coefficients, we have
\[
B_{\alpha,\alpha'} = \left\{
\begin{array}{ll}
q_{\alpha,\alpha'}, \qquad \alpha\notin\{(n,+1), (n+1,-1)\}
\\ \frac{\beta_n}{\beta_n+\lambda}, \qquad \alpha=(n,+1)
\\ \frac{\beta_{n+1}}{\beta_{n+1}-\lambda}, \qquad \alpha=({n+1},-1)
\end{array}
\right.
\]
Recall that $Q:=(q_{\alpha,\alpha'})_{\alpha,\alpha'\in E}$ is a
permutation matrix, more precisely a cyclic permutation of $E$.
The matrix $B$ is obtained by replacing two non-zero
elements of $Q$ with, respectively, $\frac{\beta_n}{\beta_n+\lambda}$ and
$\frac{\beta_{n+1}}{\beta_{n+1}-\lambda}$. The characteristic polynomial of $B$
is therefore equal to
\[
p(t)=t^{2N}-\frac{\beta_n}{\beta_n+\lambda} \cdot
\frac{\beta_{n+1}}{\beta_{n+1}-\lambda}.
\]
Recall that we assume $\beta_n\leq \beta_{n+1}$.

\smallskip\noindent
{\it The case $\lambda\in \, ]-\beta_n,0[
\, \cup \, ]\beta_{n+1}-\beta_n,\beta_{n+1}[$}. In this case
\[
\beta_n\beta_{n+1}> \beta_n\beta_{n+1}+\lambda(\beta_{n+1}-
\beta_n-\lambda)=(\beta_n+\lambda)(\beta_{n+1}-\lambda)>0,
\]
and therefore $B$ has spectral radius
\[
\rho(B) = \left(\frac{\beta_n}{\beta_n+\lambda} \cdot
\frac{\beta_{n+1}}{\beta_{n+1}-\lambda}\right)^{\frac1{2N}}> 1.
\]
Let us go back to the matrix
$B^\epsilon=(B^\epsilon_{\alpha,\alpha'})_{\alpha,\alpha'\in E}$.
In this case, {\it all} non-zero elements of $B$ are modified. Indeed,
if $B_{\alpha,\alpha'}>0$, then $B^\epsilon_{\alpha,\alpha'}=C_\alpha(\lambda,\epsilon)$,
defined as in \eqref{calpha}.
Therefore, the characteristic polynomial is
\[
p^\epsilon(t)=t^{2N}-\prod_{\alpha\in E} C_\alpha(\lambda,\epsilon)=t^{2N}-F(\lambda,\epsilon).
\]
Therefore, the spectral radius of $B^\epsilon$ is
\[
\rho(B^\epsilon) = \left(F_n(\lambda,\epsilon,\cB)\right)^{\frac1{2N}}.
\]
Since $B^\epsilon$ is an irreducible matrix with non-negative entries,
by the Perron-Frobenius theorem \cite[Th. I.6.4]{asmussen},
$\rho(B^\epsilon)$ is an eigenvalue of
$B^\epsilon$ with multiplicity 1; moreover this eigenvalue
is associated with a right-eigenvector $(r_\alpha)$ and a
left-eigenvector $(l_\alpha)$
such that $r_\alpha>0$ and $l_\alpha>0$ for all $\alpha\in E$.

Since $\epsilon\to\rho(B^\epsilon)$ is strictly decreasing
with value $\rho(B)>1$ at $\epsilon=0$ and 0 limit as
$\epsilon\to+\infty$, then there exists a
unique $\epsilon_0$ such that $F_n(\lambda,\epsilon_0,\cB)=1$.
Let us set
\[
\hat{Z}_{\alpha,\ga'}(t):=\frac{e^{-\epsilon_0t}\, {Z}_{\alpha,\ga'}\, r_{\ga'}}{r_\alpha},
\qquad \hat{M}_{\alpha,\ga'}(d\tau):=\frac{e^{-\epsilon_0\tau}\, {M}_{\alpha,\ga'}(d\tau)\, r_{\ga'}}{r_\alpha}.
\]
By construction, $\hat{M}$ is a {\it semi-Markov kernel}, i.e.
\[
\sum_{\ga'\in E} \int_0^{+\infty} \hat{M}_{\alpha,\ga'}(d\tau)=1, \qquad \forall \ \alpha\in E.
\]
Moreover, $\hat{Z}_{\alpha,\ga'}(t)$ satisfies
\[
\hat Z_{\alpha,\ga'}(t) = e^{-\epsilon_0t}\, \bbP_\alpha(\tau_1>t)\, \delta_{\alpha,\ga'}
+\sum_{\ell=1}^\infty \hat M^{\ell*}_{\alpha,\ga'}([0,t]).
\]
Let now $(\hat X_k, \hat \tau_{k+1})_{k\geq 0}$ be a Markov renewal process with kernel
$\hat M$. Then we can write
\[
\hat U_{\alpha,\ga'}([0,t]) := \sum_{\ell=1}^\infty \hat M^{\ell*}_{\alpha,\ga'}([0,t])
=\sum_{\ell=1}^\infty \bbP_\alpha(\hat \tau_1+\cdots+\hat \tau_\ell\leq t, \, \hat X_\ell=\ga').
\]
Notice that the kernel $\hat M$ has finite mean:
\[
    \hat \mu := \sum_{\ga,\ga' \in E} \int_0^{+\infty} \tau \, \hat\nu_\ga \,
\hat M_{\ga,\ga'}(d\tau) \in \ (0,\infty),
\]
where $\hat \nu_{\ga'}=l_{\ga'}/\sum_\gamma l_\gamma$ is the unique invariant measure
of $\hat X$ on $E$.
Then, by the Markov Renewal theorem \cite[Th.~VII.4.3]{asmussen}
\begin{equation} \label{eq:per_renewal_theorem}
\lim_{t\to+\infty} \frac1t\, {\hat U_{\ga,\ga'}([0,t])} \;=\; \frac{\hat \nu_{\ga'}}{\hat \mu}.
\end{equation}
Therefore we obtain
\[
{Z}_{\alpha,\ga'}(t)=\bbP_\alpha(\tau_1>t)\, \delta_{\alpha,\ga'} + e^{\epsilon_0t}
\, \hat U_{\alpha,\ga'}([0,t]) \, \frac{r_\alpha}{r_{\ga'}}
\]
and summing over $\ga'\in E$
\[
\E\left(\exp\left(-\lambda J_n([0,t])\right)\right) =
\bbP_\alpha(\tau_1>t) + e^{\epsilon_0t}\, {r_\alpha} \sum_{\ga'} \frac{\hat U_{\alpha,\ga'}([0,t])}{r_{\ga'}}
\]
and therefore,
\[
f_n(\lambda,\cB)=\lim_{t\to+\infty} \frac1t \,
\log\E\left(\exp\left(-\lambda J_n([0,t])\right)\right)
=\epsilon_0.
\]

\smallskip\noindent
{\it The case  $\lambda\in\, [0,\beta_{n+1}-\beta_n]$}.
In this case we have $\rho:=\rho(B)\in \, ]0,1]$. Recall that
\[
\begin{split}
Z_{\alpha,\ga'}(t) & =
\bbP_\alpha(\tau_1>t)\, \delta_{\alpha,\ga'} +\sum_{\ell=1}^\infty M^{\ell*}_{\alpha,\ga'}([0,t])
\geq \bbP_\alpha(\tau_1>t)\, \delta_{\alpha,\ga'},
\end{split}
\]
where for $\alpha=(i,\sigma)$ and $\alpha'=(i',\sigma')$
\[
\bbP_\alpha(\tau_1>t) =
\, \left(1-e^{-{\beta_{i}}\, \frac1{2t^2}}\right)
\]
and summing over $\ga'\in E$
\[
\E\left(\exp\left(-\lambda J_n([0,t])\right)\right) \geq
\bbP_\alpha(\tau_1>t) =
\, \left(1-e^{-\frac{\beta_{i}}{2t^2}}\right),
\]
and therefore
\begin{eqnarray}\label{esti}
\nonumber
& \liminf_{t\to+\infty} \frac1t \, \log\E\left(\exp\left(-\lambda J_n([0,t])\right)\right)
\geq \liminf_{t\to+\infty} \frac1t \, \log\left(1-e^{-\frac{\beta_{i}}{2t^2}}\right)
\\ & \sim \lim_{t\to+\infty} \frac1t \, \log\left(\frac{\beta_{i}}{2t^2}\right) = 0.
\end{eqnarray}
Now, let us set
\[
\tilde{M}_{\alpha,\ga'}(d\tau):=\frac{{M}_{\alpha,\ga'}(d\tau)\, r_{\ga'}}{\rho\, r_\alpha}.
\]
By construction, $\tilde{M}$ is a {\it semi-Markov kernel}, i.e.
\[
\sum_{\ga'\in E} \int_0^{+\infty} \tilde{M}_{\alpha,\ga'}(d\tau)=1, \qquad \forall \ \alpha\in E.
\]
Moreover, ${Z}_{\alpha,\ga'}(t)$ satisfies
\[
Z_{\alpha,\ga'}(t) = \bbP_\alpha(\tau_1>t)\, \delta_{\alpha,\ga'}
+\sum_{\ell=1}^\infty \rho^{-\ell}\, \tilde M^{\ell*}_{\alpha,\ga'}([0,t]).
\]
Let now $(\tilde X_k, \tilde \tau_{k+1})_{k\geq 0}$ be a Markov renewal process with kernel
$\tilde M$. Then we can write
\[
\tilde U_{\alpha,\ga'}([0,t]) := \sum_{\ell=1}^\infty \tilde M^{\ell*}_{\alpha,\ga'}([0,t])
=\sum_{\ell=1}^\infty \bbP_\alpha(\tilde \tau_1+\cdots+\tilde \tau_\ell\leq t, \, \tilde X_\ell=\ga').
\]
Notice that the kernel $\tilde M$ has finite mean:
\[
    \tilde \mu := \sum_{\ga,\ga' \in E} \int_0^{+\infty} \tau \, \tilde\nu_\ga \,
\tilde M_{\ga,\ga'}(d\tau) \in \ (0,\infty),
\]
where $\tilde \nu_{\ga'}$ is the unique invariant measure
of $\tilde X$ on $E$. Indeed, by \eqref{haydn},
$M_{\ga,\ga'}(d\tau)\sim C\, \tau^{-3}$ as $\tau\to+\infty$.
Then, by the Markov Renewal theorem \cite[Th.~VII.4.3]{asmussen}
\begin{equation} \label{eq:per_renewal_theorem2}
\lim_{t\to+\infty} \frac1t\, {\tilde U_{\ga,\ga'}([0,t])} \;=\; \frac{\tilde \nu_{\ga'}}{\tilde \mu}.
\end{equation}
Therefore we obtain
\[
Z_{\alpha,\ga'}(t) \leq \, \delta_{\alpha,\ga'} +
\, \tilde U_{\alpha,\ga'}([0,t]) \, \frac{r_\alpha}{r_{\ga'}}
\]
and summing over $\ga'\in E$
\[
\E\left(\exp\left(-\lambda J_n([0,t])\right)\right) \leq 1 +
\sum_{\ga'\in E} \tilde U_{\alpha,\ga'}([0,t]) \, \frac{r_\alpha}{r_{\ga'}} \leq C(1+t),
\]
for some constant $C>0$. Hence
\begin{equation}\label{esti2}
\limsup_{t\to+\infty} \frac1t \, \log\E\left(\exp\left(-\lambda J_n([0,t])\right)\right)
\leq \limsup_{t\to+\infty} \frac1t \, \log\left(C(1+t)\right)  = 0.
\end{equation}
By \eqref{esti} and \eqref{esti2} we obtain
\[
f_n(\lambda,\cB)=\lim_{t\to+\infty} \frac1t \,
\log\E\left(\exp\left(-\lambda J_n([0,t])\right)\right)=0.
\]

\smallskip\noindent
{\it Continuity and convexity of $f_n(\cdot,\cB)$}.
It is a standard fact that
\[
\frac{\partial^2}{\partial\lambda^2}\frac1t \, \log\E\left(\exp\left(-\lambda J_n([0,t])\right)\right)
=\frac1t \, \frac{\E\left(J^2_n([0,t])e^{-\lambda J_n([0,t])}\right)-\left(\E\left(J_n([0,t])e^{-\lambda J_n([0,t])}\right)\right)^2}{\left(\E\left(e^{-\lambda J_n([0,t])}\right)\right)^2}
\]
is non-negative, so by passing to the limit $t\to+\infty$,
$f_n(\cdot,\cB)$ is convex and finite and therefore continuous.

\smallskip\noindent
{\it The Gallavotti-Cohen symmetry relation}.
Equation \eqref{GC} follows from the analogous symmetry of $F$
\[
F_n(\lambda,\epsilon,\cB)=F_n(\beta_{n+1}-\beta_n-\lambda,\epsilon,\cB).
\]
\end{proof}
\begin{rem}\label{analytic}
{\rm The proof shows that the lack of analyticity of $f_n(\cdot,\cB)$ is related to the tail of the
distribution of $\tau_i$ or, equivalently, to the probability of having slow particles
in the system. Indeed, the crucial estimate \eqref{esti}, which shows that $f_n\geq 0$,
follows from the polynomial decay of the probability that a particle takes an amount
of time $t$ to reach the next scatterer, namely $\bbP(\tau_i>t)\sim t^{-2}$, $t\to+\infty$.
This is also related to the absence of spectral gap of the dynamics.
Physically, the origin of this phenomenon is the fact that the particle may get an arbitrarily small speed
(taking thus an arbitrarily large amount of time before the next collision)
which prevents the system from converging exponentially fast to the stationary state.   This should not be
regarded as an artefact of the model but rather as a general feature of collisional dynamics.
We leave a more complete study of the convergence to the invariant measure to future work.
}

\end{rem}

\subsection{Green-Kubo formula and fluctuation-dissipation relations}
\label{fluctu}

It is well-known that the Gallavotti-Cohen symmetry \eqref{GC} implies the
identification of the thermal conductivity with the variance of the time-integrated
current, a relation known as the Green-Kubo formula.  See for instance \cite{reybellet1,reybellet2}
for a derivation of this fact in the context of chains of anharmonic oscillators,
where the analyticity of the cumulant generating function is proven under suitable
hypothesis on the interaction potentials.

In fact, the Gallavotti-Cohen symmetry \eqref{GC} implies a relation between certain
partial derivatives of the cumulant generating function $f_n$, which is correct if the
generating function is smooth for $\lambda$ close to $0$. The identification of the
partial derivatives of $f_n$ with the physical quantities of interest, namely the
thermal conductivity and the variance of the current, requires an exchange of limits
which is often hard to justify.

In our models, as remarked before Proposition \ref{mgf},
the cumulant generating function is {\it not} smooth
for $\lambda$ close to $0$ and $\beta_n\ne\beta_{n+1}$ and therefore the computations
which are usually performed require some care. We are nevertherless able to obtain
the desired results, by considering only left-derivatives when needed and by
showing directly that the physical quantities of interest satisfy the expected relations.
Notice that we shall denote for
any function $g:\, ]-\varepsilon,\varepsilon[\mapsto \R$
\[
g(0^-):=\lim_{t\uparrow 0}g(t), \qquad g(0^+):=\lim_{t\downarrow 0}g(t),
\]
whenever any of such limit exists.

\begin{prop}\label{myt} Let $\cB=(\beta_0,\ldots,\beta_N)\in\R_+^{N+1}$.
If $0<\beta_n<\beta_{n+1}$ then
\[
\frac{\partial f_n}{\partial \lambda} \, (0^-,\cB)= -\frac{T_{n}-T_{n+1}}{Z_N}
=-\cJ_n,
\]
where $T_i=\beta_i^{-1}$, see \eqref{cJ}.
\end{prop}
\begin{proof} We consider $\lambda\in\, ]-\beta_n,0[$. Then by Proposition \ref{mgf},
$f_n(\lambda,\cB)>0$ is defined by the relation $F_n(\lambda,f_n(\lambda,\cB),\cB)=1$.
Hence by the implicit function Theorem
\[
\frac{\partial f_n}{\partial \lambda}(\lambda,\cB)
= - \frac{\partial F_n}{\partial \lambda}(\lambda,f_n(\lambda,\cB),\cB)
\left/ \frac{\partial F_n}{\partial\epsilon}(\lambda,f_n(\lambda,\cB),\cB)\right. .
\]
A computation yields
\[
\begin{split}
 & \frac{\partial F_n}{\partial \lambda}(\lambda,\epsilon,\cB) = \prod_{\alpha\in E\backslash
\{(n,1),(n+1,-1)\}} C_\alpha(0,\epsilon) \quad \cdot
 \\ & \cdot \beta_n\beta_{n+1}\int_{\R_+^2} v_1\, v_2\,
 \frac12\left(v_2^2-v_1^2\right)
\exp\left(-\frac{\epsilon}{v_1} -\frac{\epsilon}{v_2}-\left(\beta_n+\lambda\right)\, \frac{v_1^2}2
-\left(\beta_{n+1}-\lambda\right)\, \frac{v_2^2}2\right) dv_1 \,dv_2,
\end{split}
\]
\[
\begin{split}
 & \frac{\partial F_n}{\partial \epsilon}(\lambda,\epsilon,\cB) \\
& = \sum_{\alpha=(i,\sigma)\in E} \beta_i\int_{\R_+} v\,
 \frac1{v} \exp\left(-\frac{\epsilon}{v} -\left(\beta_i+\Delta_n(\alpha)\lambda\right)\frac{v^2}2
\right) dv \ \prod_{\alpha'\in E\backslash
\{\alpha\}} C_{\alpha'}(\lambda,\epsilon).
\end{split}
\]
Since $f_n(0,\cB)=0$, by letting $\lambda\uparrow 0$ we find
\[
\begin{split}
\frac{\partial f_n}{\partial \lambda}(0^-,\cB) & =
\frac{\beta_n\beta_{n+1}\int_{\R_+^2} v_1 \, v_2 \,
 \frac12\left(v_2^2-v_1^2\right)
e^{-\beta_n\frac{v_1^2}2
-\beta_{n+1}\frac{v_2^2}2} \, dv_1 \, dv_2}
{\sum_{\alpha=(i,\sigma)\in E} \beta_i\int_{\R_+}
\exp\left(-\beta_i\frac{v^2}2\right) dv}
\\ & =\frac{\frac1{\beta_{n+1}}-\frac1{\beta_n}}{\sum_{\alpha=(i,\sigma)\in E}
\sqrt{\frac{\pi\beta_i}2}} = \frac{T_{n+1}-T_n}{Z_N}.
\end{split}
\]
\end{proof}
\noindent
We consider now the equilibrium case $\beta_0=\cdots=\beta_N=\beta>0$.
\begin{prop}\label{bern1} Let $\cB_{\rm eq}:=(\beta,\ldots,\beta)\in\R_+^{N+1}$, $\beta>0$.
Then the function $]-\beta,\beta[\, \ni\lambda\mapsto f_n(\lambda,\cB_{\rm eq})$
is analytic, even and
\[
\frac{\partial f_n}{\partial \lambda} \, (0,\cB_{\rm eq})=0, \qquad
\frac{\partial^2 f_n}{\partial \lambda^2} \, (0,\cB_{\rm eq})=\frac1N\sqrt{\frac{2}{\pi\beta^5}}.
\]
\end{prop}
\begin{proof} The relation $f_n(\lambda,\cB_{\rm eq})=f_n(-\lambda,\cB_{\rm eq})$ follows from
the Gallavotti-Cohen symmetry \eqref{GC}. The analyticity follows from the
implicit function theorem.

Since $T_n=T_{n+1}$, we know by Proposition \ref{myt} that $\frac{\partial f_n}{\partial \lambda} \, (0,\cB_{\rm eq})=0$.
We consider $\lambda\in\, ]-\beta,0[$. Then by Proposition \ref{mgf},
$f_n>0$ is defined by the relation $F_n(\lambda,f_n(\lambda,\cB_{\rm eq}),\cB_{\rm eq})=1$.
Hence, differentiating this relation w.r.t. $\lambda$ and then setting
$\lambda=0$ we obtain, since $f_n(0,\cB_{\rm eq})=\frac{\partial f_n}{\partial \lambda} \, (0,\cB_{\rm eq})=0$,
\[
\frac{\partial^2 f_n}{\partial \lambda^2}(0,\cB_{\rm eq})
= - \frac{\partial^2 F_n}{\partial \lambda^2}(0,0,\cB_{\rm eq})
\left/ \frac{\partial F_n}{\partial\epsilon}(0,0,\cB_{\rm eq})\right. .
\]
Now,
\[
\begin{split}
& \frac{\partial^2 F_n}{\partial \lambda^2}(0,0,\cB_{\rm eq}) =
 \beta^2\int_{\R_+^2} v_1\, v_2\,
 \frac14\left(v_2^2-v_1^2\right)^2
\exp\left(-\frac\beta2\left(v_1^2+v_2^2\right)\right) dv_1 \,dv_2
\\ & = \frac{\beta^2}4\int_{\R_+^2} v_1\, v_2
 \left(v_2^4+v_1^4-2v_1^2v_2^2\right)
e^{-\frac\beta2\left({v_1^2}+{v_2^2}\right) } dv_1 \,dv_2=
\frac{\beta^2}4\left( 2\,\frac1\beta\, \frac8{\beta^3}
-2\, \left(\frac2{\beta^2}\right)^2\right)=\frac2{\beta^2},
\end{split}
\]
\[
\begin{split}
\frac{\partial F_n}{\partial \epsilon}(0,0,\cB_{\rm eq}) =
 - 2N\beta\int_{\R_+} v\,
 \frac1{v}\exp\left(-\beta\, \frac{v^2}2\right) dv = -N\sqrt{2\pi\beta},
\end{split}
\]
so that
\[
\frac{\partial^2 f_n}{\partial \lambda^2}(0,\cB_{\rm eq})=\frac2{\beta^2} \cdot \frac1N\sqrt{\frac1{2\pi\beta}}
=\frac1N\sqrt{\frac{2}{\pi\beta^5}}.
\]
\end{proof}

\noindent Let $n\in\{0,\ldots,N-2\}$ and for $0\leq\Delta\beta<\beta$, let us set
\[
\beta_i:=\beta>0, \quad i\notin\{n,n+1\}, \qquad \beta_n:=\beta-\frac{\Delta\beta}2,
\quad \beta_{n+1}:=\beta+\frac{\Delta\beta}2,
\]
and let us set $\cB(\Delta\beta):=(\beta_0,\ldots,\beta_N)$. Notice that $\cB(0)=\cB_{\rm eq}$.
Let us set
\[
A:=\left\{(\lambda,\Delta\beta): \, \Delta\beta\in[0,\beta[, \ \lambda\in
\left]-\beta,\beta-{\Delta\beta}\right[ \, \right\}
\]
and $g_n:A\mapsto\R_+$
\[
g_n(\lambda,\Delta\beta):=f_n\left(\lambda,\cB(\Delta\beta)\right),
\quad \forall \ (\lambda,\Delta\beta)\in \,A.
\]
\begin{prop}\label{bern3}
We have the Green-Kubo relation
\[
\frac{\partial^2 g_n}{\partial \lambda^2} (0,0) =
-2\, \frac{\partial}{\partial\Delta\beta}\, \frac{\partial g_n}{\partial \lambda} (0,0)
= \frac1N\sqrt{\frac{2}{\pi\beta^5}}.
\]
\end{prop}
\begin{proof}
Notice first that $g_n(\lambda,0)=f_n(\lambda,0)$ is analytic around $\lambda=0$, so that
\[
\frac{\partial^2 g_n}{\partial \lambda^2} (0,0) = \frac{\partial^2 f_n}{\partial \lambda^2} (0,\cB_{\rm eq}) =
\frac1N\sqrt{\frac{2}{\pi\beta^5}}
\]
by Proposition \ref{bern1}. Now, for fixed $\Delta\beta>0$, we have
by Proposition \ref{myt}
\[
\begin{split}
\frac{\partial g_n}{\partial \lambda} (0,\Delta\beta)& =-\frac{\beta_{n}^{-1}-\beta_{n+1}^{-1}}{Z_N(\cB)}
\\ & = -\frac{\Delta\beta}{\beta_{n}\beta_{n+1}} \sqrt{\frac2\pi}\left((2N-4)\sqrt{\beta}+2\sqrt{\beta-\Delta\beta}
+2\sqrt{\beta+\Delta\beta}\right)^{-1}.
\end{split}
\]
Therefore
\[
 \frac{\partial}{\partial\Delta\beta}\, \frac{\partial g_n}{\partial \lambda} (0,0)=
\frac1N\sqrt{\frac{1}{2\pi\beta^5}}
\]
and the result is proven.
\end{proof}

\begin{rem}{\rm
It is possible to prove directly that at equilibrium, i.e. $\beta_0=\cdots=\beta_N=\beta>0$,
\[
\lim_{t\to+\infty}\frac1t\E\left(\left(J_n([0,t])\right)^2\right) = \frac1N\sqrt{\frac{2}{\pi\beta^5}}
=\frac{\partial^2 f_n}{\partial \lambda^2} \, (0,\cB_{\rm eq}),
\]
which shows that the formal exchange of limits in $t\to+\infty$ and in $\lambda\to0^-$ in the formula
\[
\frac{\partial^2}{\partial\lambda^2}\frac1t \, \log\E\left(\exp\left(-\lambda J_n([0,t])\right)\right)
\]
yields indeed a correct result. Using the Gallavotti-Cohen symmetry one can prove both Proposition
\ref{bern3} and the further equality
\[
\frac{\partial^2 g_n}{\partial \lambda^2} (0,0) =
-2\, \frac{\partial}{\partial \lambda}\, \frac{\partial g_n}{\partial\Delta\beta} (0,0)=
-2\, \frac{\partial}{\partial\Delta\beta}\, \frac{\partial g_n}{\partial \lambda} (0,0)
= \frac1N\sqrt{\frac{2}{\pi\beta^5}}.
\]
}
\end{rem}

\section{Confined tracers}\label{confined}

\subsection{Generalities and physical observables.}
In this section, we introduce a model which gives rise to a qualitatively different behavior for the self-consistent temperature profile.  The self-consistent temperature profile of the scatterers in the wandering tracers model was linear.  We will see that in the case of confined tracers, the temperature profile becomes non-linear.  A major difference between the two models is the dependence of the thermal conductivity on the set of temperatures of the scatterers.   For an arbitrary temperature distribution of the scatterers, we have seen that in the case of wandering tracers, the conductivity was identified with a frequency of collisions of a tracer with two neighboring scatterers.  As such and because the wandering tracer travels through the whole system, it was dependent on the temperature of every scatterer.  In the case of confined tracers, the conductivity is a purely local function of the set of temperatures.

\noindent The general structure of the process is again one of a Markov renewal process, the notations and proofs are strictly analogous to the case of wandering tracers.
In this model, the disposition of the scatterers is the same but there are exactly $N$ tracer particles locked in between the scatterers, including the ones on the boundaries. The $n$-th particle moves in between the scatterers, in the interval $I_n=[n-1,n]$, being reflected at the scatterers $w_{n-1}$ and $w_n$ with a random velocity $p$ distributed according to
\begin{equation}
\phi^{\pm}_{\beta_n}(p)=p^{\pm}\beta_n e^{-{\beta_n} \frac{p^2}{2}}.
\end{equation}
Because the particle is reflected, the sign in the distribution is the opposite of the sign of the incoming velocity.  Those models are  described by $N$ independent Markov renewal processes.  Each scatterer exchange energy with its two adjacent tracer particles and in order to express the self-consistency condition, we must  introduce notations to describe the motion of each tracer. We describe now the process describing the motion of the $n$-th particle traveling between scatterers $w_n$ and $w_{n+1}$.  The state space of the Markov chain is $E=\{-1,+1\}$, with transition
probability defined by $q_{1,-1}=q_{-1,1}=0$.

Let $(q_{n,0},p_{n,0})$ the initial data and velocity of the $n$-th particle.  We define $\sigma_{n,0}={\rm sign}(p_{n,0})$.
We consider now the Markov chain $(\sigma_{n,k})_{k\geq 0}$ in $E$ with initial
state $X_0=\sigma_{n,0}$. In fact, the Markov chain has a deterministic evolution $\sigma_{n,k}=(-1)^k\sigma_{n,0}$, $k\geq 0$.

For each $\sigma\in E$, we write $\hat\sigma=\hf(\sigma+1)$.
Then the time of the first collision with a scatterer is
\[
S_{n,0}=S_{n,0}(q_{n,0},p_{n,0}):=\frac{n+\hat\sigma_{n,0}-q_{n,0}}{p_{n,0}}>0,
\]
We now define the time the particle takes between two
subsequent visits to the scatterers.
Conditionally on the $\sigma$-algebra generated by $(\sigma_{n,k})_{k\geq 0}$
the sequence $(\tau_{n,k})_{k\geq 1}$
is independent with distribution defined by
\begin{equation}\label{taucon}
 \bbP(\tau_{n,k}\in d\tau \, | \, \sigma_{n,k-1})
= \frac{\beta_{n+\hat\sigma_{n,k-1}}}{\tau^3} \,
\exp\left(-\frac{\beta_{n+\hat\sigma_{n,k-1}}}{2\tau^2}\right)
\, \bo{(\tau>0)} \, d\tau,
\end{equation}
recall \eqref{psibeta}.
The time of the $k$-th collision with one of the two scatterers $w_n$ and $w_{n+1}$ is
\[
S_{n,k}:=S_{n,0}+\tau_{n,1}+\cdots+\tau_{n,k}, \qquad k\geq 1.
\]
Before time $S_{n,0}$, the particle moves with uniform velocity
$p_{n,0}$. Between time $S_{n,k-1}$ and time $S_{n,k}$, the particle
moves with uniform velocity $\frac{\sigma_{n,k}}{\tau_{n,k}}$ and
$(S_{n,k})_{k\geq 0}$ is the sequence of times when $q_{n,t}\in
\{n,n+1\}$. In particular we define the sequence of incoming velocities
$v_{n,k}$ at time $S_{n,k}$
\begin{equation}\label{vcon}
v_{n,0}:=p_{n,0}, \qquad v_{n,k}:=\frac{\sigma_{n,k}}{\tau_{n,k}}, \quad k\geq 1.
\end{equation}
We define the stochastic process $(q_{n,t},p_{n,t})_{t\geq 0}$
with values in $[n,n+1]\times \R^*$
\begin{equation}\label{qpcon}
(q_t,p_t) :=
\left\{
\begin{array}{ll}
(q_{n,0}+p_{n,0}t,p_{n,0}) \qquad {\rm if} \quad  t<S_{n,0},
\\ \\
\left( n+\hat\sigma_{n,k-1}
+\frac{\sigma_{n,k}}{\tau_{n,k}}(t-S_{n,k-1}),
\frac {\sigma_{n,k}}{\tau_{n,k}}
\right)  \ \ {\rm if} \ \
S_{n,k-1}\leq t< S_{n,k}, \ k\geq 1,
\end{array}
\right.
\end{equation}
Then, in analogy with Propositions \ref{markov}, \ref{invmeas}, \ref{markovren} and \ref{markovren2},
we have the result,

{\prop{The process $((q_{n,t},p_{n,t})_{t>0})_{0\leq n\leq N-1}$ is Markov and its only invariant measure is given by
\begin{equation}
\mu(\un p,\un q)=\frac1{\hat Z_N}\prod_{n=0}^{N-1}\bo{I_n}(q_n)
\left[\bo{(p_n>0)} \, \beta_{n} \, e^{-\beta_{n}\frac{p^2_n}{2}}+
\bo{(p_n<0)} \, \beta_{n+1}\, e^{-\beta_{n+1}\frac{p^2_n}{2}}\right]
\label{selfmu2}
\end{equation}
where $\beta_0=\beta_L$ and $\beta_{N}=\beta_R$ and ${Z_N}$ is the
normalization constant,
\begin{equation}
{\hat Z_N}:=\prod_{n=0}^{N-1}Z_n,
\qquad Z_n:=\left(
\frac{\pi\beta_{n}}{2}\right)^\hf+
\left(\frac{\pi\beta_{n+1}}{2}\right)^\hf.
\end{equation}
}}
\noindent We next identify the physical quantities of interest. The energy
exchanged between the scatterer $n$ and its two neighboring particles
during a time interval $[0,t]$ is given by
\begin{eqnarray*}
E_{n}([0,t]):=&&\hf\sum_{k\geq 0:\, S_{n,k}\leq t}\left(v^2_{n,k+1}-v^2_{n,k}\right)\bo{(\hat\sigma_{n,k}=0)}\nonumber\\&+&\hf\sum_{k\geq 0:\, S_{n-1,k}\leq t}\left(v^2_{n-1,k+1}-v^2_{n-1,k}\right)
\bo{(\hat\sigma_{n-1,k}=1)}\nonumber,
\end{eqnarray*}
recall that, by \eqref{vcon} and \eqref{qpcon}, $v_{n,k}$ and $v_{n,k+1}$
are respectively the incoming
and the outcoming velocity of the $n$-th particle at time $S_{n,k}$.
The energy exchanged between scatterers $n$ and $(n+1)$
during a time interval $[0,t]$ is given by
\[
J_{n\to n+1}([0,t]):=\hf\sum_{k\geq 1:\, S_{n,k}\leq t}
v^2_{n,k} \, \sigma_{n,k}
\]
The total entropy flow $S_n([0,t])$ and $S([0,t])$
due to the exchange of energy between the
scatterers and a particle can be defined as in \eqref{entropyflow}.
The energy flow per unit time $\cE_n$ in the stationary state,
the entropy flow per unit time $\cS_n$ and $\cS$,
and the average current of energy per unit time $\cJ_n$ between $w_n$
and $w_{n+1}$, can be defined as in, respectively, \eqref{en}, \eqref{sn}
and \eqref{jn}.

As in the case of wandering tracers (Proposition \ref{limit}), we may study the above limits defining the physical properties of the model.  As compared to Proposition \ref{limit}, the main difference resides in the expression of  the energy exchanged $\cE_n$ with the system.  This is the origin of the difference of shapes of the temperature profiles of the two models.

\begin{prop} 
For all $n=1,\ldots,N-1$,
\begin{equation}\label{cE2}
\cE_n= \frac{T_n-T_{n-1}}{Z_{n-1}}+\frac{T_n-T_{n+1}}{Z_{n}}\quad{\rm and}\quad
\cE_0=\frac{T_0-T_1}{Z_0},\quad \cE_N=\frac{T_N-T_{N-1}}{Z_{N-1}},
\end{equation}
\begin{equation}\label{cJ2}
\cJ_n= \frac{T_n-T_{n+1}}{Z_n},  \qquad \cS=
\sum_{n=0}^{N-1}\frac{(T_n-T_{n+1})^2}{Z_nT_nT_{n+1}}\geq 0.
\end{equation}
\end{prop}
\noindent The proof is completely analogous to that of Proposition
(\ref{limit}) and we do not repeat it.
The main feature of the proof is again that by using the renewal theorem, the conductivity
\[
\kappa_n=\frac{\cJ_n}{T_n-T_{n+1}}=\frac{1}{Z_n}
\]
appears as a frequency of collision of the tracer with the walls of the box
to which it is confined.

\subsection{Self-consistency condition, temperature profile and Fourier's law.}
We now derive the consequence of the self-consistency condition $\cE_n=0$, $n=1,\ldots,N-1$ on the shape of the temperature profile.

\noindent We set
\[
g_N(x):=\sum_{i=0}^{N-1} \bo{\left[\frac i{N},\frac {i+1}{N}\right[}(x) \, N(T_{i+1}-T_i)
\]
and
\[
h_N(x):=T_L+\int_0^x g_N(t)\, dt.
\]
Notice that $h_N(i/N)=T_i$ and that $h_N$ linearly interpolates between these values.
{\prop{ {\bf (Self-consistency condition)}
The only collection $(T_n)_{n=0,\ldots,N}$
such that
$$
 \cE_n=0,\; n=1,\ldots,N-1
$$
with $T_0=T_L$ and $T_N=T_R$, is the solution of
\begin{equation}
\left (\frac{(T_n-T_{n+1})}{(T_n)^{-\hf}+(T_{n+1})^{-\hf}}+\frac{(T_n-T_{n-1})}{(T_n)^{-\hf}+(T_{n-1})^{-\hf}}\right )=0,
\qquad 1\leq n\leq N-1.
\label{selfT}
\end{equation}
In this case, when $N\rightarrow+\infty$, $h_N$
converges uniformly to the function
\begin{equation}
h(x):= \left( T_L^{\frac{3}{2}}+
x(T_R^{\frac{3}{2}}-T_L^{\frac{3}{2}})\right)^\frac{2}{3}, \qquad x\in[0,1],
\label{solution}
\end{equation}
unique solution of the
equation
\[
\left\{ \begin{array}{ll} \left(h^\hf \, h'\right)'=0, \quad x\in]0,1[,
\\ \\
h(0)=T_L,\ h(1)=T_R
\end{array}\right.
\]
}}
\noindent {\it Proof.}
We first note that (\ref{selfT}) implies that $h_N$ is a solution of the elliptic equation
\[
\int_0^1 a_N \, h_N'\, \varphi' \, dx=0, \qquad \forall \, \varphi\in C^\infty_c(0,1),
\]
where
\[
\begin{split}
a_N (x) & := \sum_{i=0}^{N-1} \bo{\left[\frac i{N},\frac {i+1}{N}\right[}(x)
\ \frac1{T_i^{-\hf}+T_{i+1}^{-\hf}} \
\frac1{(h_N(\lfloor xN\rfloor/N))^{-\hf}+(h_N(\lceil xN\rceil/N))^{-\hf}}.
\end{split}
\]
This can be seen by writing,
\begin{eqnarray}
\int_0^1 a_N \, h_N'\, \varphi' \, dx&=&\sum_{i=0}^{N-1}N\frac{T_{i+1}-T_i}{T_i^{-\hf}+T_{i+1}^{-\hf}}\int_{\frac{i}{N}}^{\frac{i+1}{N}}\varphi'(x)\,dx\nonumber \\ &=&\sum_{i=0}^{N-1}N\frac{T_{i+1}-T_i}{T_i^{-\hf}+T_{i+1}^{-\hf}}\left(\varphi\left(\frac{i+1}{N}\right)-\varphi\left(\frac{i}{N}\right)\right),
\end{eqnarray}
summing by parts this last equation and using (\ref{selfT}).
We are going to show below  that the sequence of functions $g_N$ is
bounded in $L^2(0,1)$.  As every subsequence is also bounded, one can
extract from every subsequence a  subsubsequence, weakly converging to
some $g\in L^2(0,1)$.
Correspondingly, from every subsequence of $h_N$, one can thus extract
a subsubsequence converging uniformly (in $x$).   Since the convergence
is uniform, every limit function $h$ obtained in this way satisfies
\[
0=\int_0^1 a_N \, h_N'\, \varphi' \, dx \to \int_0^1 2\, h^{\hf} \, h'\, \varphi' \, dx,
\]
i.e.
\[
\left( h^{\hf} \, h' \right)' = 0.
\]
with boundary conditions $h(0)=T_L$, $h(1)=T_R$.  This equation admits a unique solution given by (\ref{solution}).  This implies that all the subsubsequences constructed above converge to (\ref{solution}) and thus that the full sequence $h_N$ converges to it as well.
We show now that the sequence $g_N$ is bounded in $L^2(0,1)$ norm.
Multiplying \eqref{selfT} by $T_n$ and summing by parts, we obtain
\[
\sum_{n=1}^{N-2} \frac{(T_n-T_{n+1})^2}{T_n^{-\hf}+T_{n+1}^{-\hf}} =
T_{N}\frac{T_{N}-T_{N-1}}{T_{N}^{-\hf}+T_{N-1}^{-\hf}} - T_1\frac{T_1-T_0}{T_1^{-\hf}+T_0^{-\hf}}
\]
Let us suppose that $T_R=T_{N+1}\geq T_0=T_L$. Then, because (\ref{selfT}) implies that the sign of the sequence $(T_{n+1}-T_n)$ is constant, we have $T_{n+1}\geq T_n$ for all $n=0,\ldots, N-1$,
and in particular $T_0\leq T_n\leq T_N$. One also obtains that
\[
T_{n+1}-T_n = \frac{T_n^{-\hf}+T_{n+1}^{-\hf}}{T_n^{-\hf}+T_{n-1}^{-\hf}}\, (T_n-T_{n-1}) \leq T_n-T_{n-1},
\]
so that the function $T$ is increasing and concave. In particular,
\[
(N-2)(T_{N}-T_{N-1})^2 \leq \sum_{n=1}^{N-2} \frac{(T_n-T_{n+1})^2}{T_n^{-\hf}+T_{n+1}^{-\hf}} \leq
T_{N-1}\frac{T_{N}-T_{N-1}}{T_{N}^{-\hf}+T_{N-1}^{-\hf}} \leq T_L^{\frac32} (T_{N}-T_{N-1}),
\]
and therefore $T_{N}-T_{N-1}\leq T_L^{\frac32}/(N-2)$. In particular,
\begin{equation}\label{eest}
N\sum_{n=1}^{N-2} (T_n-T_{n+1})^2 \leq \frac1{2T_R^{\frac32}}
\sum_{n=1}^{N-2} \frac{(T_n-T_{n+1})^2}{T_n^{-\hf}+T_{n+1}^{-\hf}} \leq
\left(\frac{T_L}{T_R}\right)^\frac32
\end{equation}
By \eqref{eest} we have that
\[
\int_0^1 g_N^2 \, dx = N\sum_{n=0}^{N} (T_n-T_{n+1})^2 \leq C
\]
which shows that the sequence $g_N$ is  bounded in $L^2(0,1)$.
\penalty-20\null\hfill$\square$\par\medbreak
\noindent  We define the local thermal conductivity by
\[
\kappa(x)=\lim_{N\to\infty}\kappa_{\lfloor N x\rfloor}
\]
for $x\in[0,1]$.
Its numerical value is easy to obtain,
\[
\lim_{N\to\infty}\kappa_{\lfloor N x\rfloor}=\lim_{N\to\infty}Z^{-1}_{\lfloor N x\rfloor}=
\lim_{N\to\infty}\left[\left(\frac{\pi\beta_{\lfloor N x\rfloor}}{2}\right)^\hf+
\left(\frac{\pi\beta_{\lfloor N x\rfloor+1}}{2}\right)^\hf\right]^{-1}=\sqrt{\frac{h(x)}{2\pi}},
\]
where $h$ is the function (\ref{solution}).
Unlike \eqref{GK0}, in this case the local conductivity has
a spatial dependence.

\subsection{Cumulant generating function and
the Gallavotti-Cohen symmetry relation}\label{confinedGC}
We can repeat the analysis of the cumulant generating function of subsection \ref{wanderingGC} in the context of the confined tracers.   The proofs are actually simpler and the results translate word by word.

\section{Conclusions and prospects.}
In the tracers-scatterers models introduced in this paper, we were able to show the validity of  Fourier's law and interpret the thermal conductivity as the collision frequency between tracers and scatterers.  This comes naturally as a consequence of the renewal theorem for Markov renewal processes.  We have recovered the two types of temperature profiles observed in deterministic systems described by local collisional dynamics. We were also able to study in details the cumulant generating function of the time-integrated current, showing in particular its lack of analyticity.  We have provided a formula allowing the computation of cumulants of any order in and out of equilibrium. In particular we have shown the validity of the Green-Kubo formula in our models.

Natural problems to study in future works are the large deviations properties and convergence to the invariant measure of those models.  They may be also further extended and studied in different interesting ways.  One possible extension is to consider cases where the tracers are transmitted or reflected according to some non-trivial probability ditribution.  In that case, the Markov chain associated to the Markov renewal process becomes non-deterministic.  Although we have provided an explicit form for the invariant measure covering that case too, it would be interesting to study in details the dynamical properties of those systems.  In particular, it appears that breaking the right-left symmetry in the transmission-reflection rules induces a modification of the stationary current which enters into competition with the current driven by the temperature gradient imposed in the system.  
It could be also interesting  to consider models, as in \cite{ColletEckmann}, where the velocity of the tracer remains unaffected by the scatterers with some positive probability.

\noindent{\bf Acknowledgments.} R.L. is supported by the French ANR network Limites Hydrodynamiques et M\'ecanique Statistique Hors Equilibre.

\end{document}